\documentclass[11pt,a4paper]{article}
\usepackage[utf8]{inputenc}
\usepackage{authblk}
\usepackage{amsmath,amsfonts,amssymb}
\usepackage{mathtools}
\usepackage{enumerate}
\usepackage{float}
\usepackage{bm}
\usepackage{upgreek}
\usepackage{rotating}
\usepackage{multicol}
\usepackage{array,multirow}
\usepackage{subfigure}
\usepackage{authblk}
\usepackage{textcomp}
\usepackage{graphicx}
\usepackage{caption}
\usepackage{algorithm}
\usepackage[noend]{algpseudocode}
\usepackage{url,hyperref}
\usepackage{amsthm}
\usepackage[left=2cm,right=2cm,top=2cm,bottom=2cm]{geometry} 

\title {Enhancing accuracy for solving American CEV model with high-order compact scheme and adaptive time stepping}

\date{}
\numberwithin{equation}{section}
\newtheorem{lemma}{Lemma}

\author[1]{Chinonso Nwankwo \thanks{Corresponding author: \url{chinonso.nwankwo@ucalgary.ca
}}}
\author[2]{Weizhong Dai}
\author[1]{Tony Ware}
\affil[1]{Department of Mathematics and Statistics, University of Calgary, Calgary T2N 1N4, Canada}
\affil[2]{Mathematics and Statistics, Louisiana Tech University, Ruston, LA 71272, USA}

\newtheorem*{remark}{Remark}
\begin{document}
\maketitle

\begin{abstract}
\noindent  In this research work, we propose a high-order time adapted scheme for pricing a coupled system of fixed-free boundary constant elasticity of variance (CEV) model on both equidistant and locally refined space-grid. The performance of our method is substantially enhanced to improve irregularities in the model which are both inherent and induced. Furthermore, the system of coupled PDEs is strongly nonlinear and involves several time-dependent coefficients that include the first-order derivative of the early exercise boundary. These coefficients are approximated from a fourth-order analytical approximation which is derived using a regularized square-root function. The semi-discrete equation for the option value and delta sensitivity is obtained from a non-uniform fourth-order compact finite difference scheme. Fifth-order 5(4) Dormand-Prince time integration method is used to solve the coupled system of discrete equations. Enhancing the performance of our proposed method with local mesh refinement and adaptive strategies enables us to obtain highly accurate solution with very coarse space grids, hence reducing computational runtime substantially. We further verify the performance of our methodology as compared with some of the well-known and better-performing existing methods.\\

\end{abstract}

\section{Introduction}\label{sec1} 
\noindent Constant elasticity of variance (CEV) model was first developed by Black \cite{Black, Black1, Chan} for the model that involves stochastic volatility. It was further studied extensively by Cox \cite{Cox}. Let $S_t$ represent an asset price. Under risk neutral probability $\mathbb{Q}$, the stochastic differential equation representing the CEV model is driven by geometric Brownian motion as 
\begin{equation}
\mathrm{d}S_t = \mu(t,S_t)\mathrm{d}t+\bar{\sigma}(t,S_t)\mathrm{d}B_t,  \quad \mu(t,S_t) = S_t,  \quad  \bar{\sigma}(t,S_t) = \sigma S^{\alpha+1}
\end{equation}
where $\lbrace B_t,t\geq 0 \rbrace$ is a standard Brownian motion, $\mu$ represent the drift, and $\sigma$ represent the volatility.
Let us consider  non-dividend-paying put options $V(t,S)$ written on this underlying asset with price $S(t)$ strike price $K$ and time to maturity $T$, then the free boundary partial differential equation is given as 
\begin{align}\label{f1} 
\frac{\partial V(t,S)}{\partial t} + \frac{\sigma S^{2\alpha +2}}{2} \frac{\partial^2 V(t,S)}{\partial S^2} + rS \frac{\partial V(t,S)}{\partial S}-rV(t,S)=0, \qquad S>f_{b}(t).
\end{align} 
Here, we solve backward in time and consider $\tau = T-t$ for which we obtain
\begin{align}\label{f2} 
\frac{\partial V(\tau,S)}{\partial \tau} - \frac{\sigma S^{2\alpha +2}}{2}  \frac{\partial^2 V(\tau,S)}{\partial S^2} - rS \frac{\partial V(\tau,S)}{\partial S}+rV(\tau,S)=0, \qquad S>f_{b}(\tau);
\end{align}
with the following initial and boundary conditions
\begin{equation}\label{f2ghi1}
V(\tau,S)=K-S, \quad \frac{\partial V(\tau,S)}{\partial S} = -1, \qquad S<f_{b}(\tau)
\end{equation}
where $f_{b}(\tau)$ represents the optimal exercise boundary which is approximated simultaneously with the option value and hedge sensitivities. The boundary and the initial conditions are further given below
\begin{equation}\label{f2ghi2}
V(\tau,f_{b}(\tau))=K-f_{b}(\tau), \quad \frac{\partial V(\tau,f_{b}(\tau))}{\partial S} = -1;
\end{equation}
\begin{equation}\label{f2ghi3}
V(\tau,\infty)=0, \quad \frac{\partial V(\tau,\infty)}{\partial S} = 0,
\end{equation}
\begin{equation}\label{f2ghi}
V(0,S)=(K-S)^+.
\end{equation}

\noindent Several researchers have proposed differential and integral methods for solving CEV model under both European and American style frameworks. Under European framework, Wei et al. \cite{Wei}  solved fuzzy  constant elasticity of variance model. Lo et al. \cite{Lo} solved the CEV model with time-dependent parameters.  Elliott et al. \cite{Elliott} described methodology for solving CEV model under regime switching. Aboulaich et al. \cite{Aboulaich} proposed an approximate numerical method for pricing of options for the constant elasticity of variance (CEV) diffusion model.  Bock et al. \cite{Bock} solved the European CEV model under stochastic volatility. Fan et al. \cite{Fan} proposed a numerical method for solving a dynamic CEV model. Chan and Ng \cite{Chan} develop Black-Scholes formula for pricing fractional CEV model.   \\

\noindent In the American style context, a few of authors have proposed proposed numerical methods for solving the CEV model. Wong and Zhao \cite{Wongb} solved the American CEV model using an artificial boundary method. They further implemented the Laplace-Carson transform \cite{Wongb} for solving such options. Pun and Wong \cite{Pun} proposed an asymptotic solution method based on perturbation approach for solving the American style CEV model. Lee \cite{Lee} solved the free boundary CEV model using a combination of intermediate function, second-order central finite difference scheme in space, and second-order Crank-Nicholson time integration method.  Nicholls and Sward \cite{Nicholls} used the discontinuous Galerkin and Runge-Kutta time integration methods to solve the American put options under CEV model. Cruz and Dias \cite{Cruz} and Lee et al. \cite{Lee0} used integral representation of the early exercise boundary for solving the American options CEV model. Jo et al. \cite{Jo} studied the convergence of Laplace inversion for solving the American put option under the CEV model. Yin and Tan \cite{Yin} studied the American put option with transaction costs under the CEV model. Ruas et al. \cite{Ruas} proposed a numerical method for pricing and static hedging of American-style options under the jump to default extended CEV (JDCEV) model. The American style JDCEV model was also considered in the work of Nunes \cite{Nunes}. Ballestra and Cecere \cite{Ballestra} extended the method of Barone-Adesi and Whaley \cite{Barone} for pricing American options under CEV model. Thakoor et al. \cite{Thakoor} proposed a high-order compact scheme for pricing the American style CEV model.\\

\noindent Our main objective in this research is to approximate a system of coupled American style CEV models consisting of option value and delta sensitivity simultaneously with the optimal exercise boundary using varying elasticity $\alpha$. To this end, we pose our solution framework as a fixed free boundary problem by considering the following transformations
\begin{equation}\label{sfbc}
E = \frac{K}{S_0},  \qquad U(\tau,x) = \frac{V(\tau,S)}{S_0},  \qquad x = \ln\left(\frac{S}{s_f(\tau)}\right),   \qquad W(\tau,x) = \frac{\partial U(\tau,x)}{\partial x}.
\end{equation}
Here, because of the inherent challenge we observed with CEV model due to the elasticity parameter $\alpha$, we scale the domain with $S_0$ to ensure consistent approximation of the option value for a given asset value $S_0$. We may also refer the reader to the work of Nicholls and Sward \cite{Nicholls}. With the transformation in (\ref{sfbc}), two system of fixed-free boundary CEV model consisting of option value and delta sensitivity can be written as follows:
\begin{equation}\label{f50}
\frac{\partial U(\tau,x)}{\partial \tau} -\xi_1(\tau,x; \alpha) \frac{\partial^2 U(\tau,x)}{\partial x^2} - \xi_2(\tau,x;\alpha)W(\tau,x)+rU(\tau,x)=0,  \quad x>0; 
\end{equation}
\begin{equation}\label{f51}
\frac{\partial W(\tau,x)}{\partial t} - \xi_1(\tau,x; \alpha)\frac{\partial^2 W(\tau,x)}{\partial x^2} - \xi_3(\tau,x; \alpha)\frac{\partial^2 U(\tau,x)}{\partial x^2}+\xi_4(\tau,x; \alpha)W(\tau,x)=0\,,\quad x>0;
\end{equation}
\begin{equation}\label{f60} 
\xi_1(\tau,x;\alpha) = \frac{\sigma^2 [e^x s_f(\tau)]^{2\alpha}}{2}, \qquad \xi_2(\tau,x;\alpha) = r+\frac{s'_{f}(\tau)}{s_{f}(\tau)}-\frac{\sigma^2 [e^x s_f(\tau)]^{2\alpha}}{2}; 
\end{equation}
\begin{equation}\label{f600} 
\xi_3(\tau,x;\alpha) = \xi_2(\tau,x;\alpha)+ \sigma^2 \alpha[e^x s_f(\tau)]^{2\alpha}, \qquad \xi_4(\tau,x;\alpha) = r+\sigma^2 \alpha[e^x s_f(\tau)]^{2\alpha};
\end{equation}
\begin{equation}\label{f61}
U(\tau,x)=E-e^x s_{f}(\tau) ,\qquad W(\tau,x)=-e^x s_{f}(\tau), \qquad x<0,
\end{equation}
\begin{equation}\label{f7}
U(0,x)=0, \qquad U(\tau,0)=E-s_{f}(\tau),  \qquad U(\tau,\infty)=0,  \qquad \tau>0;
\end{equation}
\begin{equation}\label{f8}
W(0,x)=0, \qquad W(\tau,0)=-s_{f}(\tau),  \qquad W(\tau,\infty)=0,  \qquad  \tau>0.
\end{equation}
\begin{remark}
It is worth mentioning that an appropriate scaling of the option value is very paramount for an efficient numerical approximation of the CEV model. If the option value and the optimal exercise boundary are not scaled appropriately, some of the coefficients of the differential operators in the model which depend on the asset price can be greatly impacted on based on the value of the elasticity $\alpha$. By scaling, we intend to circumvent this challenge.
\end{remark}
\noindent It is easy to observe that in most of cases, the derivative operator has a time and space-dependent co-efficients; hence, presenting  strongly nonlinear model. Moreover, $\xi_2(\tau,x;\alpha)$ and $\xi_3(\tau,x;\alpha)$ involves the first-order derivative of the optimal exercise and the latter is discontinuous at payoff. In this research, we first present a fourth-order analytical approximation based on a square-root for which the time-dependents $\xi_2(\tau,x;\alpha)$ and $\xi_3(\tau,x;\alpha)$ are both improved and approximated. Next, we derive a third-order Robin boundary compact scheme for obtaining a discrete solution of the option value at the left boundary. The interior is discretized with fourth-order compact finite difference scheme. For time integration, we employ 5(4) Dormand-Prince embedded pairs \cite{Dormand} which has consistently shown its efficiency for solving the free boundary options pricing problem as seen in the work of Nwankwo and Dai \cite{Nwankwoo, Nwankwos}. Hence, an efficient fourth-order in space and fifth-order in time numerical scheme for approximating the option value is established. Furthermore, the optimal exercise boundary is computed from the left boundary solution of the option value based on the value matching condition. The left boundary of the delta sensitivity is in turn obtained from the early exercise boundary based on the smooth-pasting condition. The interior of the delta sensitivity is then finally computed using the fourth-order compact finite difference scheme and 5(4) Dormand-Prince Runge-Kutta embedded pairs. \\ 

\noindent The remaining sections of this work are organised as follows: In Section \ref{sec2}, we describe our proposed high-order numerical adaptive implementation for solving a system of free boundary CEV model. In section \ref{sec3}, we carry out numerical experiment and report our numerical results and findings. We further verify the performance of our method by comparing with some of the existing methods and conclude our study in \ref{sec4}.  

\section{Numerical Methods} \label{sec2}
The numerical methodology for solving the system of free boundary CEV model consisting of the option value and delta sensitivity is presented in this section. Our numerical method solves the option value, delta sensitivity, and optimal exercise boundary simultaneously. The main computational domain for our numerical method is given as $(\tau,x) \in ([0,T] \times [0,\infty])$. However, we truncate the unbounded domain and introduce an artificially impose right boundary $x_M$. This is because of the rapid vanishing effect of the option value and delta sensitivity as we move out of the money. The grid description for our numerical method is presented below:
\begin{equation}
M=\frac{x_M}{h}, \quad x_i \in [0,x_M], \quad i = 0,1, \cdots M, \quad k_n = \tau_{n+1}-\tau_n,  \quad \left|[0,T]\right| = \sum_{n=0}^{N-1} k_n.
\end{equation}
For convenience, we label the numerical approximations of the option price and delta sensitivity,and the optimal exercise boundary as $u^n, w^n$, and $s_{f}^n$, respectively.
\subsection{Fourth-order uniform and staggered analytical approximations}\label{sdfn}
The high-order analytical approximation from which the time dependent coefficient $\xi_2(\tau,x;\alpha)$ and $\xi_3(\tau,x;\alpha)$ are obtained from a more transformation based on the square-root function as 
\begin{equation}
Q(\tau,x)= \sqrt{U(\tau,x)-E+e^x s_f(\tau)},  \qquad x\geq 0.
\end{equation}
It implies that
\begin{equation}\label{lbsf}
U(\tau,x) = Q^2(\tau,x)+E-e^xs_f(\tau).
\end{equation}
First introduced by Kim et al. \cite{Kim1}, some authors have implemented this transformation for solving and improving degeneracy, smoothness and convergence rate in American style options \cite{Lee, Nwankwoss}. Here, we will first obtain the derivatives of the square-root function and then use an extrapolated Taylor series expansion on both uniform and staggered grids to derive an analytical approximation for obtaining the first-order derivative of the optimal exercise boundary. To this end, we first take the derivatives of the value function based on the relation in (\ref{lbsf}) as follows:
\begin{equation}\label{lbsf0}
\frac{\partial U(\tau,x)}{\partial x} = 2Q(\tau,x)\frac{\partial Q(\tau,x)}{\partial x} +E-e^xs_f(\tau),  \quad \frac{\partial U(\tau,x_0)}{\partial x} = -s_f(\tau);
\end{equation}
\begin{equation}\label{lbsf1}
\frac{\partial U(\tau,x)}{\partial \tau} = 2Q(\tau,x)\frac{\partial Q(\tau,x)}{\partial \tau}-e^xs_f(\tau),  \quad \frac{\partial U(\tau,x_0)}{\partial \tau} = -s'_f(\tau);
\end{equation}
\begin{equation}\label{lbsf2}
\frac{\partial^2 U(\tau,x)}{\partial x^2} =2Q(\tau,x)\frac{\partial^2 Q(\tau,x)}{\partial x^2}  +2\left(\frac{\partial Q(\tau,x)}{\partial x}\right)^2-e^xs_f(\tau),
\end{equation}
\begin{equation}\label{lbsf3}
\frac{\partial^2 U(\tau,x_0)}{\partial x^2} = 2\left(\frac{\partial Q(\tau,x_0)}{\partial x}\right)^2-s_f(\tau).
\end{equation}
Let $\beta = 2\alpha$. Substituting the derivative of the value function to the option value PDE in (\ref{f50}) at $x_0=0$, we obtain that
\begin{equation}
-s'_f(\tau)-\frac{\sigma^2}{2}s_f^{\beta}(\tau) \left[2\left(\frac{\partial Q(\tau,x_0)}{\partial x}\right)^2-s_f(\tau)\right]+\left[r-\frac{\sigma^2}{2}s_f^{\beta}(\tau)-\frac{s'_f(\tau)}{s_f(\tau)}\right]s_f(\tau)+r[E-s_f(\tau)].
\end{equation}
Simplifying further, we then obtain that
\begin{equation}
\frac{\partial Q(\tau,x_0)}{\partial x} = \frac{\sqrt{rE}}{\sigma s_f^{\beta/2}(\tau)}.
\end{equation}
To obtain the second-order derivative of the square-root, we further take the derivatives of the option value as follows:
\begin{equation}
\frac{\partial^3 U(\tau,x)}{\partial x^3} =2Q(\tau,x)\frac{\partial^3 Q(\tau,x)}{\partial x^3}  +6\frac{\partial Q(\tau,x)}{\partial x}\frac{\partial^2 Q(\tau,x)}{\partial x^2}-e^xs_f(\tau),
\end{equation}
\begin{equation}
\frac{\partial^3 U(\tau,x_0)}{\partial x^3} =6\frac{\partial Q(\tau,x_0)}{\partial x}\frac{\partial^2 Q(\tau,x_0)}{\partial x^2}-s_f(\tau),
\end{equation}
\begin{equation}\label{lbsf5}
\frac{\partial^2 U(\tau,x)}{\partial x\partial \tau} = 2Q(\tau,x)\frac{\partial^2 Q(\tau,x)}{\partial \tau^2} +2\frac{\partial Q(\tau,x)}{\partial x}\frac{\partial Q(\tau,x)}{\partial \tau}-e^xs'_f(\tau).
\end{equation}
Note that
\begin{equation}
\frac{\partial^2 U(\tau,x_0)}{\partial x\partial \tau} = -s'_f(\tau),
\end{equation}
which implies that
\begin{equation}
\frac{\partial Q(\tau,x_0)}{\partial \tau} = 0.
\end{equation}
\noindent Substituting the derivatives of the option value to the delta sensitivity PDE in (\ref{f51}) at $x_0 = 0$, we further obtain 
\begin{align}
-s'_f(\tau)-&\frac{\sigma^2}{2}s_f^{\beta}(\tau) \left[6\frac{\partial Q(\tau,x_0)}{\partial x}\frac{\partial^2 Q(\tau,x_0)}{\partial x^2}-s_f(\tau)\right]\nonumber
\\&\nonumber\\ & +\left[r-\frac{\sigma^2}{2}s_f^{\beta}(\tau)+\frac{s'_f(\tau)}{s_f(\tau)}+\beta \frac{\sigma^2}{2}s_f^{\beta}(\tau)\right]\left[2\left(\frac{\partial Q(\tau,x_0)}{\partial x}\right)^2-s_f(\tau)\right]\nonumber \\&\nonumber\\ & -\left(r+\beta \frac{\sigma^2}{2}s_f^{\beta}(\tau)\right)s_f(\tau).
\end{align}
Simplifying further, we then obtain
\begin{equation}
\frac{\partial^2 Q(\tau,x_0)}{\partial x^2} = -\frac{\beta \sqrt{rE}}{3\sigma s_f^{\beta /2}(\tau)}-\frac{2\xi_2(\tau,x;\alpha)\sqrt{rE}}{3\sigma^3 s_f^{3\beta/2}(\tau)}.
\end{equation}
\begin{remark}
In our previous research works \cite{Nwankwoo, Nwankwoss}, we derived the derivative of the square-root function up to its third-order derivative. This is still possible here. However, for the sake of simplicity, in our implementation for this work, our derivation will be limited to the second-order derivative of the square-root function in obtaining the fourth-order analytical approximation of the time-dependent coefficients. We observed that using our local space grid refinement procedure and adaptive time stepping which will be presented in the following subsection for this section, a high level solution accuracy will be maintained using up to the second-order derivative of the square root function.
\end{remark}
\noindent Because our grid domain will be discretized with fourth-order compact finite difference scheme, we need to ensure that $\xi_2(\tau,x;\alpha)$ and $\xi_3(\tau,x;\alpha)$ are approximated with fourth-order accuracy. To this end, we present a fourth-order extrapolated Taylor series expansion around $x_0=0$ by considering the following Lemma.
\begin{lemma}
Assume that $u(\tau,x)\in C^{1,2}((0,T],[0,x_{max}])$, it holds
\begin{align}\label{dnr33lpp}
-\frac{415}{72}Q(\tau,x_0=0)&+8Q(\tau,\bar{x})-3Q(\tau,2\bar{x})+\frac{8}{9}Q(\tau,3\bar{x})-\frac{1}{8}Q(\tau,4\bar{x})=\nonumber
\\&\nonumber\\ & \frac{25}{6}\bar{x}Q'(\tau,x_0)+\bar{x}^2 Q''(\tau,x_0)+O(h^6),  \quad \bar{x} << x.
\end{align}
\end{lemma}
\begin{proof}
With an extrapolated Taylor series expansion, it is easy to derive (\ref{dnr33lpp}). Hence, we skip the derivation. For similar derivation, we refer the readers to the work of Nwankwo and Dai \cite{Nwankwoo, Nwankwos}.
\end{proof}
\noindent Note that $Q(\tau,x_0)=0$. Substituting the derivatives of the square-root function in (\ref{dnr33lpp}), we obtain
\begin{equation}
\frac{s'_f(\tau)}{s_f(\tau)} = \frac{\mathcal{M}_2(\tau;\alpha)-\mathcal{M}_3(\tau;\alpha)+\mathcal{M}_4(\tau)}{\mathcal{M}_1(\tau;\alpha)}+O(h^4),
\end{equation}
\begin{equation}
\mathcal{M}_1(\tau;\alpha) = -\frac{2 \bar{x}^2 \sqrt{rE}}{3 \sigma^3 s_f^{3\beta /2}(\tau)}, \qquad \mathcal{M}_2(\tau;\alpha) = \frac{\beta \bar{x}^2\sqrt{rE}}{3 \sigma s_f^{\beta /2}(\tau)}+\frac{2 \nu \bar{x}^2 \sqrt{rE}}{3 \sigma^3 s_f^{3\beta /2}(\tau)};
\end{equation}
\begin{equation}
\mathcal{M}_3(\tau;\alpha) = \frac{25\bar{x}\sqrt{rE}}{6\sigma s_f^{\beta/2}(\tau)},  \qquad \mathcal{M}_4(\tau) =8Q(\tau,\bar{x})-3Q(\tau,2\bar{x})+\frac{8}{9}Q(\tau,3\bar{x})-\frac{1}{8}Q(\tau,4\bar{x});
\end{equation}
\begin{equation}
\xi_2(\tau,x;\alpha) = r-\frac{\sigma^2 [e^xs_f(\tau)]^\beta}{2}+\frac{s'_f(\tau)}{s_f(\tau)}+O(h^4).
\end{equation}
\begin{equation}
\xi_3(\tau,x;\alpha) = \xi_2(\tau,x;\alpha)+ \sigma^2 \alpha[e^x s_f(\tau)]^{\beta}+O(h^4).
\end{equation}
Another valuable alternative would be to consider a staggered fourth-order boundary scheme similar to the staggered sixth-order boundary scheme presented in the work of Nwankwo and Dai \cite{Nwankwoss}. The benefit of this implementation is to ensure that we are using more of the node points close to the left boundary for approximating these time-dependent coefficients. Moreover, it will enable the boundary scheme to be more generalized and non-uniform suitable for the implementation of the non-uniform high-order numerical scheme. To this end, we consider another Lemma below for the staggered case.
\begin{lemma}
Assume that $u(\tau,x)\in C^{1,2}((0,T],[0,x_{max}])$, it holds that
\begin{align}\label{dnr33lpq}
b_0(h)Q(\tau,x_0)&+b_1(h)Q(\tau,h_1)+b_2(h)Q(\tau,h_2)+b_3(h)Q(\tau,h_3)+b_4(h)Q(\tau,h_4)=\nonumber
\\&\nonumber\\ & c_0(h)Q'(\tau,x_0)+d_0(h) Q''(\tau,x_0),  \qquad h_i = \gamma_i h;
\end{align}
with
\begin{equation}
b_0(h) = a_8(h) [a_6(h) (a_3(h)-1)-(a_4(h)-1)]-[a_7(h) (a_4(h)-1)-(a_5(h)-1)],
\end{equation}
\begin{equation}
b_1(h) = a_8(h)a_6(h)a_3(h),  
\end{equation}
\begin{equation}
b_2(h) = a_6(h)a_8(h)+a_4(h)a_8(h)+a_4(h)a_7(h),  
\end{equation}
\begin{equation}
b_3(h) =a_5(h)+a_7(h)+a_8(h),  \quad  b_4(h) =1;
\end{equation}
\begin{equation}
c_0(h) = a_8(h) [a_6(h) (a_3(h) h_1-h_2 )-(a_4(h) h_2-h_3 )]-[a_7(h) (a_4(h) h_2-h_3 )-(a_5(h) h_3-h_4 )],
\end{equation}
\begin{equation}
d_0(h) = \frac{a_8(h)}{2} [a_6(h) (a_3(h) h_1^2-h_2^2 )-(a_4(h) h_2^2-h_3^2 )]-\frac{1}{2} [a_7(h) (a_4(h) h_2^2-h_3^2 )-(a_5(h) h_3^2-h_4^2 )],
\end{equation}
\begin{equation}
a_3(h)=\frac{h_2^5}{h_1^5},  \quad a_4(h) =\frac{h_3^5}{h_2^5}, \quad a_5(h)=\frac{h_4^5}{h_3^5};
\end{equation}
\begin{equation}
a_6(h)=\frac{a_4(h) h_2^4-h_3^4}{a_3(h) h_1^4-h_2^4}, 
\end{equation}
\begin{equation}
a_7(h) =\frac{a_5(h) h_3^4-h_4^4}{a_4(h) h_2^4-h_3^4 },
\end{equation}
\begin{equation}
a_8(h)=\frac{a_7(h) (a_4(h) h_2^3-h_3^3 )-(a_5(h) h_3^3-h_4^3)}{a_6(h) (a_3(h) h_1^3-h_2^3 )-(a_4(h)h_2^3-h_3^3 ) }.
\end{equation}
\end{lemma}
\begin{proof}
The derivation is a simpler version to the one presented in the work of Nwankwo and Dai \cite{Nwankwoss} for a staggered sixth-order boundary scheme. Hence, we refer the reader to the work of Nwankwo and Dai \cite{Nwankwoss} and skip the detail derivation of (\ref{dnr33lpq}).
\end{proof}
\noindent The time dependent coefficients which involve the first-order derivative of the optimal exercise boundary is then obtained from the staggered scheme as follows:
\begin{equation}
\frac{s'_f(\tau)}{s_f(\tau)} = \frac{\mathcal{M}_6(\tau;\alpha)-\mathcal{M}_7(\tau;\alpha)+\mathcal{M}_8(\tau)}{\mathcal{M}_5(\tau;\alpha)}, \quad \mathcal{M}_5(\tau;\alpha) = -\frac{2 d_0(h) \sqrt{rE}}{3 \sigma^3 s_f^{3\beta /2}(\tau)};
\end{equation}
\begin{equation}
 \mathcal{M}_6(\tau;\alpha) = \frac{\beta d_0(h)\sqrt{rE}}{3 \sigma s_f^{\beta /2}(\tau)}+\frac{2 \nu d_0(h) \sqrt{rE}}{3 \sigma^3 s_f^{3\beta /2}(\tau)}, \quad \mathcal{M}_7(\tau;\alpha) = \frac{c_0(h)\sqrt{rE}}{\sigma s_f^{\beta/2}(\tau)};
\end{equation}
\begin{equation}
 \mathcal{M}_8(\tau) =b_1(h)Q(\tau,h_1)+b_2(h)Q(\tau,h_2)+b_3(h)Q(\tau,h_3)+b_4(h)Q(\tau,h_4).
\end{equation}
Subsequently, $\xi_2(\tau,x;\alpha)$ and $\xi_3(\tau,x;\alpha)$ can also be approximated.

\subsection{Third-order compact Robin boundary closures}\label{subs22}
In this research, our implementation considers a scenario where we do not know the left boundary value of the option value and compute the latter first before the early exercise boundary. Hence, the discrete solution of the left boundary can be coupled in the discrete matrix system from which the solution of the option value $\textbf{u}^{n+1}$ is obtained. In such scenario, a Neumann or Robin left boundary scheme needs to be developed at $x_0 =0$. Once $\textbf{u}^{n+1}$ is approximated, the early exercise boundary $s_f(\tau_{n+1})$ is obtained from the discrete solution of the option value based on the value matching condition as 
\begin{equation}
s_f(\tau_{n+1}) = E-u(\tau_{n+1},x_0).
\end{equation}
In turn, the left boundary value of the delta sensitivity can be obtained from the optimal exercise boundary based on smooth pasting condition as follows:
\begin{equation}
w(\tau_{n+1},x_0) = -s_f(\tau_{n+1}).
\end{equation}
Here, we present a new fourth-order Robin boundary scheme for obtaining the discrete equation of the option value at $x_0$ by considering the following Lemma.
\begin{lemma}
Assume that $u(\tau,x)\in C^{1,2}((0,T],(0,x_{max}])$, it holds that
\begin{align}\label{dnr33lp66}
\frac{5}{3}u''(\tau,x_0)+\frac{2}{3}u''(\tau,x_1)-\frac{1}{3}u''(\tau,x_2) &=-\frac{1}{h^2}\left(7u(\tau,x_0)-8u(\tau,x_1)+u(\tau,x_2)\right)\nonumber\\&\nonumber\\&-\frac{6}{h}u'(\tau,x_0)+O(h^3).
\end{align}
\end{lemma}
\begin{proof}
To derive (\ref{dnr33lp66}), we consider an extrapolated Robin boundary condition presented in the work of Yun et al. \cite{Yun} as given below
\begin{align}\label{ehoib2f}
\frac{35}{3}u''(\tau,x_0)+\frac{8}{3}u''(\tau,x_1)-\frac{1}{3}u''(\tau,x_2) &= -\frac{1}{12h^2}\left(31u(\tau,x_0)-32u(\tau,x_1)+u(\tau,x_2)\right)\nonumber\\&\nonumber\\&-\frac{30}{h}u(\tau,x_0)-2hu'''(\tau,x_0) +O(h^4).
\end{align}
Substituting $u'''(\tau,x_0)$ in (\ref{ehoib2f}) with
\begin{align}\label{ehoib2g}
2hu'''(\tau,x_0) = 10u''(\tau,x_0)+2u''(\tau,x_1)-\frac{24}{h^2}\left(u(\tau,x_1)-u(\tau,x_0)\right)+\frac{24}{h}u'(\tau,x_0)+O(h^3).
\end{align}
We can obtain (\ref{dnr33lp66})
\end{proof}
\begin{remark}
Note that the scheme in (\ref{ehoib2g}) is a well-known Neumann compact boundary scheme. One may see the work of Cao et al. \cite{Cao} and Yun et al. \cite{Yun} for further detail regarding the scheme in (\ref{ehoib2g}).
\end{remark}
\noindent Considering the value matching and smooth pasting conditions
\begin{equation}
u(\tau,x_0) = E-s_f(\tau),  \qquad w(\tau,x_0) = -s_f(\tau);
\end{equation}
we can generate another Robin boundary condition given below
\begin{equation}
w(\tau,x_0)-u(\tau,x_0) = -E.
\end{equation}
Substituting to (\ref{dnr33lp66}), we then obtain
\begin{align}\label{drl6d}
\frac{5}{3}u''(\tau,x_0)+\frac{2}{3}u''(\tau,x_1)&-\frac{1}{3}u''(\tau,x_2) = -\frac{7+6h}{h^2}u(\tau,x_0)\nonumber\\&\nonumber\\&+\frac{1}{h^2}\left(8u(\tau,x_1)-u(\tau,x_2)\right) +\frac{6}{h}E+O(h^3).
\end{align}
Hence, we can describe our compact boundary scheme in (\ref{drl6d}) as a fourth-order Robin boundary scheme.
\subsection{Fourth-order uniform and non-equidistant compact interior scheme and system of semi-discrete equations}
For the interior discretization of the option value, we consider the well-known fourth-order compact finite difference as follows:
\begin{align}
u''(\tau,x_{i-1})+10u''(\tau,x_i)+u''(\tau,x_{i+1}) =\frac{12}{h^2}\left[u(\tau,x_{i-1})-2u(\tau,x_i)+u(\tau,x_{i+1})\right] +O(h^4).
\end{align}
Coupled with the fourth-order Robin boundary scheme presented in subsection \ref{subs22}, we then obtain a semi-discrete differential equation in the form of method of line (MOL) as follows:
\begin{equation}
A_u\textbf{u}_{xx} = B_u\textbf{u}+\textbf{b}_u, \quad A_{w}\textbf{w}_{xx} = B_{w}\textbf{w}+\textbf{b}_w;
\end{equation}
with
$$ A_u=
\begin{bmatrix}
\frac{5}{3} & \frac{2}{3} & -\frac{1}{3} & 0 & 0 & 0 & \cdots & 0 & 0\\
1 & 10 & 1 & 0 & 0 & 0 & \cdots & 0 & 0\\
0 & 1 & 10 & 1 & 0 & 0 & \cdots & 0& 0\\
0 & 0 & 1  & 10 & 1 & 0 & \cdots & 0 & 0\\
\vdots & \vdots & \vdots  &\ddots & \ddots & \ddots & \vdots& \vdots & \vdots\\
0 & 0 & \cdots  & 0 & 1 & 10 & 1 & 0 & 0\\
0 & 0 & \cdots  & 0 & 0 & 1 & 10 & 1& 0\\
0 & 0 & \cdots & 0 & 0 & 0 & 1 & 10 & 1\\
0 & 0 & \cdots & 0 & 0 & 0 & 0 & 1 & 10
\end{bmatrix}_{M \times M}
,  \quad \textbf{b}_u= 
\begin{bmatrix}
\frac{6}{h}K\\
0\\
0 \\
0\\
0\\
0\\
0\\
0\\
0
\end{bmatrix}_{M \times 1};$$

$$ B_u= \frac{1}{h^2}
\begin{bmatrix}
-7+6h& 8 & -1 & 0 & 0 & 0 & \cdots & 0 & 0\\
12 & -24 & 12 & 0 & 0 & 0 & \cdots & 0 & 0\\
0 & 12 & -24 & 12 & 0 & 0 & \cdots & 0& 0\\
0 & 0 & 12  & -24 & 12 & 0 & \cdots & 0 & 0\\
\vdots & \vdots & \vdots  &\ddots & \ddots & \ddots & \vdots& \vdots & \vdots\\
0 & 0 & \cdots  & 0 & 12 & -24 & 12 & 0 & 0\\
0 & 0 & \cdots  & 0 & 0 & 12 & -24 & 12 & 0\\
0 & 0 & \cdots & 0 & 0 & 0 & 12 & -24 & 12\\
0 & 0 & \cdots & 0 & 0 & 0 & 0 & 12 & -24
\end{bmatrix}_{M \times M}
,$$
\begin{equation}\label{discret1}
\textbf{u}_{\tau}^n =\bm{\xi}_1^n A_u^{-1}\left(B_u\textbf{u}^n+\textbf{b}_u^n\right)+\bm{\xi}_2^n\textbf{w}^n-r\textbf{u}^n+O(h^4).
\end{equation}
After integration of semi-discrete equation in (\ref{discret1}), the optimal exercise boundary is then obtained from the left boundary of the option value based on value matching condition as given below
\begin{equation}
u(\tau_{n+1},x_0) = E-s_f(\tau_{n+1}).
\end{equation}
\noindent The left boundary of the delta sensitivity is then obtained from the optimal exercise boundary based on smooth pasting condition as follows:
\begin{equation}
w(\tau_{n+1},x_0) = -s_f(\tau_{n+1}).
\end{equation}
For the near boundary scheme for the delta sensitivity, we consider the fourth-order combined compact scheme presented in the work of Nwankwo and Dai \cite{Nwankwos} and Abrahamsen and Fornberg \cite{Abrahamsen} as follows:
\begin{equation}\label{dnp1}
w''(\tau,x_1) = \frac{15}{2h^3}\left[u(\tau,x_2)-u(\tau,x_0)] \right)-\frac{3}{2h^2}\left[w(\tau,x_{0})+8w(\tau,x_{1})+w(\tau,x_{2})\right]+O(h^4).
\end{equation}
\noindent For the interior scheme, we use the well-known fourth-order compact finite difference scheme as
\begin{align}\label{drl6d2}
w''(\tau,x_{i-1})+10w''(\tau,x_i)+w''(\tau,x_{i+1}) =\frac{12}{h^2}\left(w(\tau,x_{i-1})-2w(\tau,x_i)+w(\tau,x_{i+1}\right) +O(h^4).
\end{align}
Coupling the left boundary values, near boundary discrete equation and the interior discretization, we obtain a semi-discrete differential equation for the delta sensitivity as follows:
$$ A_w= 
\begin{bmatrix}
10 & 0 & 0 & 0 & 0 & 0 & \cdots & 0 & 0\\
1 & 10 & 1 & 0 & 0 & 0 & \cdots & 0 & 0\\
0 & 1 & 10 & 1 & 0 & 0 & \cdots & 0& 0\\
0 & 0 & 1  & 10 & 1 & 0 & \cdots & 0 & 0\\
\vdots & \vdots & \vdots  &\ddots & \ddots & \ddots & \vdots& \vdots & \vdots\\
0 & 0 & \cdots  & 0 & 1 & 10 & 1 & 0 & 0\\
0 & 0 & \cdots  & 0 & 0 & 1 & 10 & 1 & 0\\
0 & 0 & \cdots & 0 & 0 & 0 & 1 & 10 & 1\\
0 & 0 & \cdots & 0 & 0 & 0 & 0 & 1 & 10
\end{bmatrix}_{M-1 \times M-1},
$$

$$
B_w= \frac{1}{h^2}
\begin{bmatrix}
-120 & -15 & 0 & 0 & 0 & 0 & \cdots & 0 & 0\\
12 & -24 & 12 & 0 & 0 & 0 & \cdots & 0 & 0\\
0 & 12 & -24 & 12 & 0 & 0 & \cdots & 0& 0\\
0 & 0 & 12  & -24 & 12 & 0 & \cdots & 0 & 0\\
\vdots & \vdots & \vdots  &\ddots & \ddots & \ddots & \vdots& \vdots & \vdots\\
0 & 0 & \cdots  & 0 & 12 & -24 & 12 & 0 & 0\\
0 & 0 & \cdots  & 0 & 0 & 12 & -24 & 12 & 0\\
0 & 0 & \cdots & 0 & 0 & 0 & 12 & -24 & 12\\
0 & 0 & \cdots & 0 & 0 & 0 & 0 & 12 & -24
\end{bmatrix}_{M-1 \times M-1}
,$$
$$
\textbf{b}_w= 
\begin{bmatrix}
\frac{75}{h^3}\left(u_2^n -u_0^n \right)-\frac{15}{h^2} w_0^n \\
0\\
0 \\
0\\
0\\
0\\
0\\
0\\
0
\end{bmatrix}_{M-1 \times 1}
,$$
\begin{equation}\label{discret2}
\textbf{w}_{\tau}^n =\bm{\xi}_1^n A_{w,y}^{-1}\left(B_{w,y}\textbf{w}^n+\textbf{b}_w^n\right)+\bm{\xi}_3^n A_u^{-1}\left(B_u\textbf{u}^n+\textbf{b}_u^n\right)-\bm{\xi}_4^n\textbf{w}^n+O(h^4).
\end{equation}
It is important to observe that $\bm{\xi}_1^n$, $\bm{\xi}_2^n$, $\bm{\xi}_3^n$, and $\bm{\xi}_4^n$ are vectorized due to dependency on $x$. For the compact differencing on non-equidistant grid, we consider the high-order non-uniform Hermitian scheme presented in the work of of Shukla and Zhong \cite{Shukla1} and Shukla et al. \cite{Shukla2} for discretizing interior nodes as follows:
\begin{equation}\label{shkla}
d_{i,1} f''(x_{i-1})+f''(x_i )+d_{i,3} f''(x_{i+1}) \approx e_{i,1} f(x_{i-1})+e_{i,2} f(x_i)+e_{i,3} f(x_{i+1} ),  \quad i=2, 3, \cdots M-1.
\end{equation}
Here, 
\begin{equation}
d_{i,1} = \left(\frac{h_{i+1}}{h_i+h_{i+1} }\right) \left(\frac{h_i^2+h_i^2 h_{i+1}^2-h_{i+1}^2}{h_i^2+3h_{i+1} h_i+h_{i+1}^2}\right),
\end{equation}
\begin{equation}
d_{i,3} = \left(\frac{h_{i}}{h_i+h_{i+1} }\right) \left(\frac{h_{i+1}^2+h_i^2 h_{i+1}^2-h_{i}^2}{h_i^2+3h_{i+1} h_i+h_{i+1}^2}\right),
\end{equation}
\begin{equation}
e_{i,1} = \left(\frac{h_{i+1}}{h_i+h_{i+1} }\right) \left(\frac{12}{h_i^2+3h_{i+1} h_i+h_{i+1}^2}\right),
\end{equation}
\begin{equation}
e_{i,2} = \frac{-12}{h_i^2+3h_{i+1} h_i+h_{i+1}^2}, \quad e_{i,3} = \left(\frac{h_{i}}{h_i+h_{i+1} }\right) \left(\frac{12}{h_i^2+3h_{i+1} h_i+h_{i+1}^2}\right).
\end{equation}
The major reason for considering the high-order non-uniform Hermitian scheme is to implement local refinement strategy such that we can concentrate more of the grid points in the locality of the left corner point. Coupled with the uniform or non-uniform fourth-order analytical approximations described in subsection \ref{sdfn} for computing the first-order derivative of the optimal exercise boundary, we aim to substantially improve the irregularity, discontinuity and singularity that reside at the left corner point and left boundary locations. For the near-boundary scheme, we consider a fourth-order scheme as follows:
\begin{align}\label{fder2}
14f''(x_1)-5f''(x_2)&+4f''(x_3)-f''(x_4) = \frac{12}{h_a^2}\left(f(x_0)-2f(x_1)+f(x_2)\right), \quad h_a \leq \min{h_i}.
\end{align}
For this work, we only consider the case for which $h_a = 0.25h$. For description, one may see Fig \ref{fig:Local111}. Also, for implementation in option pricing problem with a high-order compact scheme, we refer the readers to the work of Chen et al. \cite{Chen}, Lee and Sun \cite{Lee1}, and Thakoor et al. \cite{Thakoor}. Furthermore, our local refinement strategy is strategically implemented to ensure that the near boundary scheme is locally uniform even though the space-time grid is globally non-equidistant. Hence, we have
\begin{equation}
A_1\textbf{u}_{xx} = B_1\textbf{u}+\textbf{b}_{u(1)}, \quad A_2\textbf{u}_{xx} = B_2\textbf{u}+\textbf{b}_{u(2)},  \quad A_2\textbf{w}_{xx} = B_2\textbf{w}+\textbf{b}_w;
\end{equation}
with $A_1$, $B_1$, $b_{u(1)}$, $A_2$, $B_2$, $b_{u(2)}$, and $b_w$ given as:
 
$$ A_1=
\begin{bmatrix}
\frac{5}{3} & \frac{2}{3}& -\frac{1}{3} & 0 & 0 & 0 & \cdots & 0 & 0\\
d_{i,1} & 1 & d_{i,3} & 0 & 0 & 0 & \cdots & 0 & 0\\
0 & d_{3,2} & 1 & d_{3,4} & 0 & 0 & \cdots & 0& 0\\
0 & 0 & d_{4,3}  & 1 & d_{4,5  } & 0 & \cdots & 0 & 0\\
\vdots & \vdots & \vdots  &\ddots & \ddots & \ddots & \vdots& \vdots & \vdots\\
0 & 0 & \cdots  & 0 & d_{M-5,M-4} & 1 & d_{M-3,M-4} & 0 & 0\\
0 & 0 & \cdots  & 0 & 0 & d_{M-4,M-3} & 1 & d_{M-2,M-3}& 0\\
0 & 0 & \cdots & 0 & 0 & 0 & d_{M-3,M-2} & 1 & d_{M-1,M-2}\\
0 & 0 & \cdots & 0 & 0 & 0 & 0 & d_{M-1,M} & 1
\end{bmatrix}
,$$

$$
A_2= 
\begin{bmatrix}
14 & 5 & -4 & 1 & 0 & 0 & \cdots & 0 & 0\\
d_{3,2} & 1 &  d_{3,4} & & 0 & 0 & \cdots & 0& 0\\
0 & d_{4,3} & 1 &  d_{4,5  } & 0 & \cdots & 0 & 0\\
\vdots & \vdots & \vdots  &\ddots & \ddots & \ddots & \vdots& \vdots & \vdots\\
0 & 0 & \cdots  & 0 & d_{M-5,M-4} & 1 & d_{M-3,M-4} & 0 & 0\\
0 & 0 & \cdots  & 0 & 0 & d_{M-4,M-3} & 1 & d_{M-2,M-3}& 0\\
0 & 0 & \cdots & 0 & 0 & 0 & d_{M-3,M-2} & 1 & d_{M-1,M-2}\\
0 & 0 & \cdots & 0 & 0 & 0 & 0 & d_{M-2,M-1} & 1
\end{bmatrix}
,$$

$$ B_u=
\begin{bmatrix}
-\frac{24}{h_a^2} & \frac{12}{h_a^2} & 0 & 0 & 0 & 0 & \cdots & 0 & 0\\
e_{2,1} & e_{2,2} & e_{2,3} & 0 & 0 & 0 & \cdots & 0 & 0\\
0 & e_{3,2} & e_{3,3} & e_{3,4} & 0 & 0 & \cdots & 0& 0\\
0 & 0 & e_{4,3}  & e_{4,4} & e_{4,5  } & 0 & \cdots & 0 & 0\\
\vdots & \vdots & \vdots  &\ddots & \ddots & \ddots & \vdots& \vdots & \vdots\\
0 & 0 & \cdots  & 0 & e_{M-5,M-4} & e_{M-4,M-4}& e_{M-3,M-4} & 0 & 0\\
0 & 0 & \cdots  & 0 & 0 & e_{M-4,M-3} & e_{M-3,M-3} & e_{M-2,M-3}& 0\\
0 & 0 & \cdots & 0 & 0 & 0 & e_{M-3,M-2} & e_{M-2,M-2} & e_{M-1,M-2}\\
0 & 0 & \cdots & 0 & 0 & 0 & 0 & e_{M-2,M-1} & e_{M-1,M-1}
\end{bmatrix}
,$$

$$
B_{w,y}= 
\begin{bmatrix}
-\frac{120}{h_a^2} & -\frac{15}{h_a^2} & 0 & 0 & 0 & 0 & \cdots & 0 & 0\\
e_{2,1} & e_{2,2} & e_{2,3} & 0 & 0 & 0 & \cdots & 0 & 0\\
0 & e_{3,2} & e_{3,3} & e_{3,4} & 0 & 0 & \cdots & 0& 0\\
0 & 0 & e_{4,3}  & e_{4,4} & e_{4,5  } & 0 & \cdots & 0 & 0\\
\vdots & \vdots & \vdots  &\ddots & \ddots & \ddots & \vdots& \vdots & \vdots\\
0 & 0 & \cdots  & 0 & e_{M-5,M-4} & e_{M-4,M-4}& e_{M-3,M-4} & 0 & 0\\
0 & 0 & \cdots  & 0 & 0 & e_{M-4,M-3} & e_{M-3,M-3} & e_{M-2,M-3}& 0\\
0 & 0 & \cdots & 0 & 0 & 0 & e_{M-3,M-2} & e_{M-2,M-2} & e_{M-1,M-2}\\
0 & 0 & \cdots & 0 & 0 & 0 & 0 & e_{M-2,M-1} & e_{M-1,M-1}
\end{bmatrix}
,$$
$$
\textbf{b}_u= 
\begin{bmatrix}
\frac{12}{h_a^2}u_0^n -\frac{3}{h_a} \left(w_2^n -w_0^n \right)\\
0\\
0 \\
0\\
0\\
0\\
0\\
0\\
0
\end{bmatrix}
,  \quad \textbf{b}_w= 
\begin{bmatrix}
\frac{75}{h_a^3}\left(u_2^n -u_0^n \right)-\frac{15}{h_a^2} w_0^n \\
0\\
0 \\
0\\
0\\
0\\
0\\
0\\
0
\end{bmatrix}
,  \quad \textbf{b}_y= 
\begin{bmatrix}
\frac{75}{h_a^3}\left(w_2^n -w_0^n \right)-\frac{15}{h_a^2} y_0^n \\
0\\
0 \\
0\\
0\\
0\\
0\\
0\\
0
\end{bmatrix}
.$$
\noindent With the discrete matrix system above, we then obtain the semi-discrete system for the option value and hedge sensitivity similar to (\ref{discret1}) and (\ref{discret2}).
\begin{figure}[H]
    \centering
    {\includegraphics[width=0.85\textwidth]{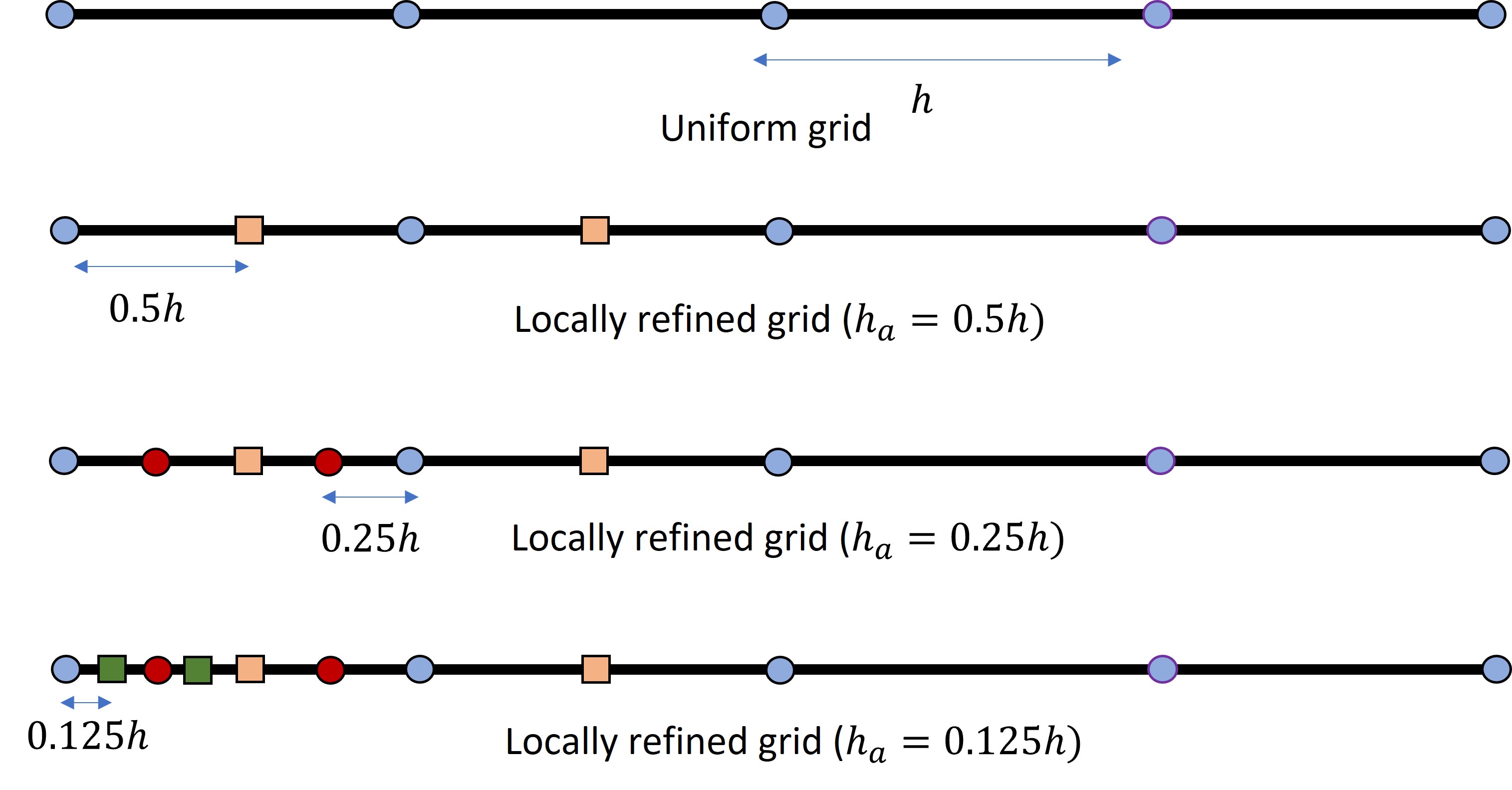}} 
    \caption{Uniform and locally refined space grids \cite{Nwankwos}.}
    \label{fig:Local111}
\end{figure}
\begin{remark}
It is very important to acknowledge the recent work of \cite{Dimitrov}, where the authors implemented a high-order non-uniform Hermitian scheme for solving the fractional Black–Scholes model. We observed that the non-uniform high-order scheme implemented by the authors in their work can also be used in this research work.
\end{remark}
\subsection{5(4) Dormand-Prince Embedded pairs}
Several Runge-Kutta low and high-order adaptive time integration based on embedded pairs exist including the few we reference here \cite{Cash, Fehlberg, Macdougall, Papageorgiou, Simos, Tsitouras, Verner0, Verner} and have been implemented for solving wide range of models \cite{Hoover, Paul, Press, Romeo, Simos0, Tremblay, Wilkie}.  In our previous work, we have implemented both low and high-order Runge-Kutta embedded pairs for solving system of coupled fixed-free boundary diffusion-convection-advection American options pricing PDEs \cite{Nwankwoo, Nwankwos}. Due to model dependency of these pairs, we further survey few 5(4) embedded pairs in our previous work and validate the importance of 5(4) Dormand-Prince and Bogacki-Shampine embedded pairs in terms of speed and accuracy. Hence, the reason for our preferred choice of 5(4) Dormand-Prince embedded pairs is to integrate two systems of fixed-free boundary PDEs in (\ref{f50})-(\ref{f51}). The adaptive nature of these embedded pairs allows us to select the optimal step at each time level instead of using a fixed time step. The description for the 5(4) Dormand-Prince embedded time integration method for solving the system of semi-discrete equations in (\ref{discret1}) and (\ref{discret1}) is described below.\\

\noindent \textbf{First stage}:
\begin{equation}
\frac{s'_f(\tau_n)}{s_f(\tau_n)} = \frac{\mathcal{M}_2(\tau_n,\bar{x};\alpha)-\mathcal{M}_3(\tau_n,\bar{x};\alpha)+\mathcal{M}_4(\tau_n,\bar{x};\alpha)}{\mathcal{M}_1(\tau_n,\bar{x};\alpha)},
\end{equation}
\begin{equation}
\bm{\xi}_1^n = \frac{\sigma^2 [e^{\textbf{x}}s_f(\tau_n)]^\beta}{2} , \qquad \bm{\xi}_2^n = r\bm{1}-\frac{\sigma^2 [e^{\textbf{x}}s_f(\tau_n)]^\beta}{2}+\frac{s'_f(\tau_n)}{s_f(\tau_n)}\bm{1};
\end{equation}
\begin{equation}
\bm{\xi}_3^n = r\bm{1}-\frac{\sigma^2 [e^{\textbf{x}}s_f(\tau_n)]^\beta}{2}+\frac{s'_f(\tau_n)}{s_f(\tau_n)}\bm{1}+\frac{\beta \sigma^2 [e^{\textbf{x}}s_f(\tau_n)]^\beta}{2},  \qquad \bm{\xi}_4^n = r\bm{1}+\frac{\beta \sigma^2 [e^{\textbf{x}}s_f(\tau_n)]^\beta}{2};
\end{equation}
\begin{equation}
\textbf{u}_{\tau}^n =\bm{\xi}_1^n A_u^{-1}\left(B_u\textbf{u}^n+\textbf{b}_u^n\right)+\bm{\xi}_2^n\textbf{w}^n-r\textbf{u}^n,
\end{equation}
\begin{equation}
\textbf{u}^{n+1/7}= \textbf{u}^n+\frac{k}{5}\textbf{u}_{\tau}^n,
\end{equation}
\begin{equation}
\textbf{w}_{\tau}^n =\bm{\xi}_1^n A_{w,y}^{-1}\left(B_{w,y}\textbf{w}^n+\textbf{b}_w^n\right)+\bm{\xi}_3^n A_u^{-1}\left(B_u\textbf{u}^n+\textbf{b}_u^n\right)-\bm{\xi}_4^n\textbf{w}^n,
\end{equation}
\begin{equation}
\textbf{w}^{n+1/7}= \textbf{w}^n+\frac{k}{5}\textbf{w}_{\tau}^n,
\end{equation}
\begin{equation}
s_{f}(\tau_{n+1/7}) = E-u(\tau_{n+1/7},x_0),  \quad w(\tau_{n+1/7},x_0) = -s_{f}(\tau_{n}).
\end{equation}
\begin{remark}
Before we proceed to the remaining Runge-Kutta stages, it is important to mention that in this research work, we denote $\bm{1}$ as a vector whose entries are 1.
\end{remark}
\noindent \textbf{Second stage}:
\begin{equation}
\frac{s'_f(\tau_{n+1/7})}{s_f(\tau_{n+1/7})} = \frac{\mathcal{M}_2(\tau_{n+1/7};\alpha)-\mathcal{M}_3(\tau_{n+1/7};\alpha)+\mathcal{M}_4(\tau_{n+1/7};\alpha)}{\mathcal{M}_1(\tau_{n+1/7})},
\end{equation}
\begin{equation}
\bm{\xi}_1^{n+1/7} = \frac{\sigma^2 [e^{\textbf{x}}s_f(\tau_{n+1/7})]^\beta}{2} , \qquad \bm{\xi}_2^{n+1/7} = r\bm{1}-\frac{\sigma^2 [e^{\textbf{x}}s_f(\tau_{n+1/7})]^\beta}{2}+\frac{s'_f(\tau_{n+1/7})}{s_f(\tau_{n+1/7})}\bm{1};
\end{equation}
\begin{equation}
\bm{\xi}_3^{n+1/7} = r\bm{1}-\frac{\sigma^2 [e^{\textbf{x}}s_f(\tau_{n+1/7})]^\beta}{2}+\frac{s'_f(\tau_{n+1/7})}{s_f(\tau_{n+1/7})}\bm{1}+\frac{\beta \sigma^2 [e^{\textbf{x}}s_f(\tau_{n+1/7})]^\beta}{2}, 
\end{equation}
\begin{equation}
\bm{\xi}_4^{n+1/7} = r\bm{1}+\frac{\beta \sigma^2 [e^{\textbf{x}}s_f(\tau_{n+1/7})]^\beta}{2},
\end{equation}
\begin{equation}
\textbf{u}_{\tau}^{n+1/7} =\bm{\xi}_1^{n+1/7} A_u^{-1}\left(B_u\textbf{u}^{n+1/7}+\textbf{b}_u^{n+1/7}\right)+\bm{\xi}_2^{n+1/7}\textbf{w}^{n+1/7}-r\textbf{u}^{n+1/7},
\end{equation}
\begin{equation}
\textbf{u}^{n+2/7}= \textbf{u}^n+\frac{3k}{40}\textbf{u}_{\tau}^n+\frac{9k}{40}\textbf{u}_{\tau}^{n+1/7},
\end{equation}
\begin{align}
\textbf{w}_{\tau}^{n+1/7} =\bm{\xi}_1^{n+1/7} A_{w,y}^{-1}\left(B_{w,y}\textbf{w}^{n+1/7}+\textbf{b}_w^{n+1/7}\right)&+\bm{\xi}_3^{n+1/7} A_u^{-1}\left(B_u\textbf{u}^{n+1/7}+\textbf{b}_u^{n+1/7}\right)\nonumber
\\&\nonumber\\
&-\bm{\xi}_4^{n+1/7}\textbf{w}^{n+1/7},
\end{align}
\begin{equation}
\textbf{w}^{n+2/7}= \textbf{w}^n+\frac{3k}{40}\textbf{w}_{\tau}^n+\frac{9k}{40}\textbf{w}_{\tau}^{n+1/7},
\end{equation}
\begin{equation}
s_{f}(\tau_{n+2/7}) = E-u(\tau_{n+2/7},x_0),  \quad w(\tau_{n+2/7},x_0) = -s_{f}(\tau_{n+2/7}).
\end{equation}
\paragraph{\textbf{Third stage}:}
\begin{equation}
\frac{s'_f(\tau_{n+2/7})}{s_f(\tau_{n+2/7})} = \frac{\mathcal{M}_2(\tau_{n+2/7};\alpha)-\mathcal{M}_3(\tau_{n+2/7};\alpha)+\mathcal{M}_4(\tau_{n+2/7};\alpha)}{\mathcal{M}_1(\tau_{n+2/7})},
\end{equation}
\begin{equation}
\bm{\xi}_1^{n+2/7} = \frac{\sigma^2 [e^{\textbf{x}}s_f(\tau_{n+2/7})]^\beta}{2} , \qquad \bm{\xi}_2^{n+2/7} = r\bm{1}-\frac{\sigma^2 [e^{\textbf{x}}s_f(\tau_{n+2/7})]^\beta}{2}+\frac{s'_f(\tau_{n+2/7})}{s_f(\tau_{n+2/7})}\bm{1};
\end{equation}
\begin{equation}
\bm{\xi}_3^{n+2/7} = r\bm{1}-\frac{\sigma^2 [e^{\textbf{x}}s_f(\tau_{n+2/7})]^\beta}{2}+\frac{s'_f(\tau_{n+2/7})}{s_f(\tau_{n+2/7})}\bm{1}+\frac{\beta \sigma^2 [e^{\textbf{x}}s_f(\tau_{n+2/7})]^\beta}{2}, 
\end{equation}
\begin{equation}
\bm{\xi}_4^{n+2/7} = r\bm{1}+\frac{\beta \sigma^2 [e^{\textbf{x}}s_f(\tau_{n+2/7})]^\beta}{2},
\end{equation}
\begin{equation}
\textbf{u}_{\tau}^{n+2/7} =\bm{\xi}_1^{n+2/7} A_u^{-1}\left(B_u\textbf{u}^{n+2/7}+\textbf{b}_u^{n+2/7}\right)+\bm{\xi}_2^{n+2/7}\textbf{w}^{n+2/7}-r\textbf{u}^{n+2/7},
\end{equation}
\begin{equation}
\textbf{u}^{n+3/7}= \textbf{u}^n+\frac{44k}{45}\textbf{u}_{\tau}^n-\frac{55k}{16}\textbf{u}_{\tau}^{n+1/7}+\frac{32k}{9}\textbf{u}_{\tau}^{n+2/7},
\end{equation}
\begin{align}
\textbf{w}_{\tau}^{n+2/7} =\bm{\xi}_1^{n+2/7} A_{w,y}^{-1}\left(B_{w,y}\textbf{w}^{n+2/7}+\textbf{b}_w^{n+2/7}\right)&+\bm{\xi}_3^{n+2/7} A_u^{-1}\left(B_u\textbf{u}^{n+2/7}+\textbf{b}_u^{n+2/7}\right)\nonumber
\\&\nonumber\\
&-\bm{\xi}_4^{n+2/7}\textbf{w}^{n+2/7},
\end{align}
\begin{equation}
\textbf{w}^{n+3/7}= \textbf{w}^n+\frac{44k}{45}\textbf{w}_{\tau}^n-\frac{55k}{16}\textbf{w}_{\tau}^{n+1/7}+\frac{32k}{9}\textbf{w}_{\tau}^{n+2/7},
\end{equation}
\begin{equation}
s_{f}(\tau_{n+3/7}) = E-u(\tau_{n+3/7},x_0),  \quad w(\tau_{n+3/7},x_0) = -s_{f}(\tau_{n+3/7}).
\end{equation}
\textbf{Fourth stage}:
\begin{equation}
\frac{s'_f(\tau_{n+3/7})}{s_f(\tau_{n+3/7})} = \frac{\mathcal{M}_2(\tau_{n+3/7};\alpha)-\mathcal{M}_3(\tau_{n+3/7};\alpha)+\mathcal{M}_4(\tau_{n+3/7};\alpha)}{\mathcal{M}_1(\tau_{n+3/7})},
\end{equation}
\begin{equation}
\bm{\xi}_1^{n+3/7} = \frac{\sigma^2 [e^{\textbf{x}}s_f(\tau_{n+3/7})]^\beta}{2} , \qquad \bm{\xi}_2^{n+3/7} = r\bm{1}-\frac{\sigma^2 [e^{\textbf{x}}s_f(\tau_{n+3/7})]^\beta}{2}+\frac{s'_f(\tau_{n+3/7})}{s_f(\tau_{n+3/7})}\bm{1};
\end{equation}
\begin{equation}
\bm{\xi}_3^{n+3/7} = r\bm{1}-\frac{\sigma^2 [e^{\textbf{x}}s_f(\tau_{n+3/7})]^\beta}{2}+\frac{s'_f(\tau_{n+3/7})}{s_f(\tau_{n+3/7})}\bm{1}+\frac{\beta \sigma^2 [e^{\textbf{x}}s_f(\tau_{n+3/7})]^\beta}{2}, 
\end{equation}
\begin{equation}
\bm{\xi}_4^{n+3/7} = r\bm{1}+\frac{\beta \sigma^2 [e^{\textbf{x}}s_f(\tau_{n+3/7})]^\beta}{2};
\end{equation}
\begin{equation}
\textbf{u}_{\tau}^{n+3/7} =\bm{\xi}_1^{n+3/7} A_u^{-1}\left(B_u\textbf{u}^{n+3/7}+\textbf{b}_u^{n+3/7}\right)+\bm{\xi}_2^{n+3/7}\textbf{w}^{n+3/7}-r\textbf{u}^{n+3/7},
\end{equation}
\begin{equation}
\textbf{u}^{n+4/7}= \textbf{u}^n+
\frac{19732k}{6561}\textbf{u}_{\tau}^n
-\frac{25360k}{2187}\textbf{u}_{\tau}^{n+1/7}+
\frac{64448k}{6561}\textbf{u}_{\tau}^{n+2/7}-\frac{212k}{729}\textbf{u}_{\tau}^{n+3/7},
\end{equation}
\begin{align}
\textbf{w}_{\tau}^{n+3/7} =\bm{\xi}_1^{n+3/7} A_{w,y}^{-1}\left(B_{w,y}\textbf{w}^{n+3/7}+\textbf{b}_w^{n+3/7}\right)&+\bm{\xi}_3^{n+3/7} A_u^{-1}\left(B_u\textbf{u}^{n+3/7}+\textbf{b}_u^{n+3/7}\right)\nonumber
\\&\nonumber\\&-\bm{\xi}_4^{n+3/7}\textbf{w}^{n+3/7},
\end{align}
\begin{equation}
\textbf{w}^{n+4/7}= \textbf{w}^n+
\frac{19732k}{6561}\textbf{w}_{\tau}^n
-\frac{25360k}{2187}\textbf{w}_{\tau}^{n+1/7}+
\frac{64448k}{6561}\textbf{w}_{\tau}^{n+2/7}-\frac{212k}{729}\textbf{w}_{\tau}^{n+3/7},
\end{equation}
\begin{equation}
s_{f}(\tau_{n+4/7}) = E-u(\tau_{n+4/7},x_0),  \quad w(\tau_{n+4/7},x_0) = -s_{f}(\tau_{n+4/7}).
\end{equation}
\textbf{Fifth stage}:
\begin{equation}
\frac{s'_f(\tau_{n+4/7})}{s_f(\tau_{n+4/7})} = \frac{\mathcal{M}_2(\tau_{n+4/7};\alpha)-\mathcal{M}_3(\tau_{n+4/7};\alpha)+\mathcal{M}_4(\tau_{n+4/7};\alpha)}{\mathcal{M}_1(\tau_{n+4/7})},
\end{equation}
\begin{equation}
\bm{\xi}_1^{n+4/7} = \frac{\sigma^2 [e^{\textbf{x}}s_f(\tau_{n+4/7})]^\beta}{2} , \qquad \bm{\xi}_2^{n+4/7} = r\bm{1}-\frac{\sigma^2 [e^{\textbf{x}}s_f(\tau_{n+4/7})]^\beta}{2}+\frac{s'_f(\tau_{n+4/7})}{s_f(\tau_{n+4/7})}\bm{1};
\end{equation}
\begin{equation}
\bm{\xi}_3^{n+4/7} = r\bm{1}-\frac{\sigma^2 [e^{\textbf{x}}s_f(\tau_{n+4/7})]^\beta}{2}+\frac{s'_f(\tau_{n+4/7})}{s_f(\tau_{n+4/7})}\bm{1}+\frac{\beta \sigma^2 [e^{\textbf{x}}s_f(\tau_{n+4/7})]^\beta}{2}, 
\end{equation}
\begin{equation}
\bm{\xi}_4^{n+4/7} = r\bm{1}+\frac{\beta \sigma^2 [e^{\textbf{x}}s_f(\tau_{n+4/7})]^\beta}{2};
\end{equation}
\begin{equation}
\textbf{u}_{\tau}^{n+4/7} =\bm{\xi}_1^{n+4/7} A_u^{-1}\left(B_u\textbf{u}^{n+4/7}+\textbf{b}_u^{n+4/7}\right)+\bm{\xi}_2^{n+4/7}\textbf{w}^{n+4/7}-r\textbf{u}^{n+4/7},
\end{equation}
\begin{equation}
\textbf{u}^{n+5/7}= \textbf{u}^n+
\frac{9017k}{3168}\textbf{u}_{\tau}^n
-\frac{355k}{33}\textbf{u}_{\tau}^{n+1/7}+
\frac{46732k}{5247}\textbf{u}_{\tau}^{n+2/7}+\frac{49k}{176}\textbf{u}_{\tau}^{n+3/7}-\frac{5103k}{18656}\textbf{u}_{\tau}^{n+4/7},
\end{equation}
\begin{align}
\textbf{w}_{\tau}^{n+4/7} =\bm{\xi}_1^{n+4/7} A_{w,y}^{-1}\left(B_{w,y}\textbf{w}^{n+4/7}+\textbf{b}_w^{n+4/7}\right)&+\bm{\xi}_3^{n+4/7} A_u^{-1}\left(B_u\textbf{u}^{n+4/7}+\textbf{b}_u^{n+4/7}\right)\nonumber
\\&\nonumber\\&-\bm{\xi}_4^{n+4/7}\textbf{w}^{n+4/7},
\end{align}
\begin{equation}
\textbf{w}^{n+5/7}= \textbf{w}^n+
\frac{9017k}{3168}\textbf{w}_{\tau}^n
-\frac{355k}{33}\textbf{w}_{\tau}^{n+1/7}+
\frac{46732k}{5247}\textbf{w}_{\tau}^{n+2/7}+\frac{49k}{176}\textbf{w}_{\tau}^{n+3/7}-\frac{5103k}{18656}\textbf{w}_{\tau}^{n+4/7},
\end{equation}
\begin{align}
s_{f}(\tau_{n+5/7}) = E-u(\tau_{n+5/7},x_0),  \quad w(\tau_{n+5/7},x_0) = -s_{f}(\tau_{n+5/7}).
\end{align}
\textbf{Sixth stage}:
\begin{equation}
\frac{s'_f(\tau_{n+5/7})}{s_f(\tau_{n+5/7})} = \frac{\mathcal{M}_2(\tau_{n+5/7};\alpha)-\mathcal{M}_3(\tau_{n+5/7};\alpha)+\mathcal{M}_4(\tau_{n+5/7};\alpha)}{\mathcal{M}_1(\tau_{n+5/7})},
\end{equation}
\begin{equation}
\bm{\xi}_1^{n+5/7} = \frac{\sigma^2 [e^{\textbf{x}}s_f(\tau_{n+5/7})]^\beta}{2} , \qquad \bm{\xi}_2^{n+5/7} = r\bm{1}-\frac{\sigma^2 [e^{\textbf{x}}s_f(\tau_{n+5/7})]^\beta}{2}+\frac{s'_f(\tau_{n+5/7})}{s_f(\tau_{n+5/7})}\bm{1};
\end{equation}
\begin{equation}
\bm{\xi}_3^{n+5/7} = r\bm{1}-\frac{\sigma^2 [e^{\textbf{x}}s_f(\tau_{n+5/7})]^\beta}{2}+\frac{s'_f(\tau_{n+5/7})}{s_f(\tau_{n+5/7})}\bm{1}+\frac{\beta \sigma^2 [e^{\textbf{x}}s_f(\tau_{n+5/7})]^\beta}{2}, 
\end{equation}
\begin{equation}
\bm{\xi}_4^{n+5/7} = r\bm{1}+\frac{\beta \sigma^2 [e^{\textbf{x}}s_f(\tau_{n+5/7})]^\beta}{2};
\end{equation}
\begin{equation}
\textbf{u}_{\tau}^{n+5/7} =\bm{\xi}_1^{n+5/7} A_u^{-1}\left(B_u\textbf{u}^{n+5/7}+\textbf{b}_u^{n+5/7}\right)+\bm{\xi}_2^{n+5/7}\textbf{w}^{n+5/7}-r\textbf{u}^{n+5/7},
\end{equation}
\begin{equation}
\textbf{u}^{n+6/7}= \textbf{u}^n+
\frac{35k}{384}\textbf{u}_{\tau}^n
+\frac{500k}{1113}\textbf{u}_{\tau}^{n+1/7}+
\frac{125 k}{192}\textbf{u}_{\tau}^{n+2/7}-\frac{2187 k}{6784}\textbf{u}_{\tau}^{n+3/7}+\frac{11 k}{84}\textbf{u}_{\tau}^{n+4/7},
\end{equation}
\begin{align}
\textbf{w}_{\tau}^{n+5/7} =\bm{\xi}_1^{n+5/7} A_{w,y}^{-1}\left(B_{w,y}\textbf{w}^{n+5/7}+\textbf{b}_w^{n+5/7}\right)&+\bm{\xi}_3^{n+5/7} A_u^{-1}\left(B_u\textbf{u}^{n+5/7}+\textbf{b}_u^{n+5/7}\right)\nonumber
\\&\nonumber\\&-\bm{\xi}_4^{n+5/7}\textbf{w}^{n+5/7},
\end{align}
\begin{equation}
\textbf{w}^{n+6/7}= \textbf{w}^n+
\frac{35k}{384}\textbf{w}_{\tau}^n
+\frac{500k}{1113}\textbf{w}_{\tau}^{n+1/7}+
\frac{125 k}{192}\textbf{w}_{\tau}^{n+2/7}-\frac{2187 k}{6784}\textbf{w}_{\tau}^{n+3/7}+\frac{11 k}{84}\textbf{w}_{\tau}^{n+4/7},
\end{equation}
\begin{align}
s_{f}(\tau_{n+6/7}) = E-u(\tau_{n+6/7},x_0),  \quad w(\tau_{n+6/7},x_0) = -s_{f}(\tau_{n+6/7}).
\end{align}
\textbf{Last stage}:
\begin{equation}
\frac{s'_f(\tau_{n+6/7})}{s_f(\tau_{n+6/7})} = \frac{\mathcal{M}_2(\tau_{n+6/7};\alpha)-\mathcal{M}_3(\tau_{n+6/7};\alpha)+\mathcal{M}_4(\tau_{n+6/7};\alpha)}{\mathcal{M}_1(\tau_{n+6/7})},
\end{equation}
\begin{equation}
\bm{\xi}_1^{n+6/7} = \frac{\sigma^2 [e^{\textbf{x}}s_f(\tau_{n+6/7})]^\beta}{2} , \qquad \bm{\xi}_2^{n+6/7} = r\bm{1}-\frac{\sigma^2 [e^{\textbf{x}}s_f(\tau_{n+6/7})]^\beta}{2}+\frac{s'_f(\tau_{n+6/7})}{s_f(\tau_{n+6/7})}\bm{1};
\end{equation}
\begin{equation}
\bm{\xi}_3^{n+6/7} = r\bm{1}-\frac{\sigma^2 [e^{\textbf{x}}s_f(\tau_{n+6/7})]^\beta}{2}+\frac{s'_f(\tau_{n+6/7})}{s_f(\tau_{n+6/7})}\bm{1}+\frac{\beta \sigma^2 [e^{\textbf{x}}s_f(\tau_{n+6/7})]^\beta}{2}, 
\end{equation}
\begin{equation}
\bm{\xi}_4^{n+6/7} = r\bm{1}+\frac{\beta \sigma^2 [e^{\textbf{x}}s_f(\tau_{n+6/7})]^\beta}{2};
\end{equation}
\begin{equation}
\textbf{u}_{\tau}^{n+6/7} =\bm{\xi}_1^{n+6/7} A_u^{-1}\left(B_u\textbf{u}^{n+6/7}+\textbf{b}_u^{n+6/7}\right)+\bm{\xi}_2^{n+6/7}\textbf{w}^{n+6/7}-r\textbf{u}^{n+6/7},
\end{equation}
\begin{equation}
\textbf{u}^{n+1} = \textbf{u}^{n+6/7}
\end{equation}
\begin{equation}
\bar{\textbf{u}}^{n+1}= \textbf{u}^n+
\frac{5179k}{57600}\textbf{u}_{\tau}^n
+\frac{7571k}{16695}\textbf{u}_{\tau}^{n+2/7}+
\frac{393 k}{640}\textbf{u}_{\tau}^{n+3/7}-\frac{92097 k}{339200}\textbf{u}_{\tau}^{n+4/7}+\frac{187 k}{2100}\textbf{u}_{\tau}^{n+5/7}+\frac{ k}{40}\textbf{u}_{\tau}^{n+6/7},
\end{equation}
\begin{align}
\textbf{w}_{\tau}^{n+5/7} =\bm{\xi}_1^{n+5/7} A_{w,y}^{-1}\left(B_{w,y}\textbf{w}^{n+5/7}+\textbf{b}_w^{n+5/7}\right)&+\bm{\xi}_3^{n+5/7} A_u^{-1}\left(B_u\textbf{u}^{n+5/7}+\textbf{b}_u^{n+5/7}\right)\nonumber
\\&\nonumber\\&-\bm{\xi}_4^{n+5/7}\textbf{w}^{n+5/7},
\end{align}
\begin{equation}
\textbf{w}^{n+1} = \textbf{w}^{n+6/7},
\end{equation}
\begin{align}
\bar{\textbf{w}}^{n+1}= \textbf{w}^n+
\frac{5179k}{57600}\textbf{w}_{\tau}^n
+\frac{7571k}{16695}\textbf{w}_{\tau}^{n+2/7}+
\frac{393 k}{640}\textbf{w}_{\tau}^{n+3/7}&-\frac{92097 k}{339200}\textbf{w}_{\tau}^{n+4/7}\nonumber\\&\nonumber\\
&+\frac{187 k}{2100}\textbf{w}_{\tau}^{n+5/7}+\frac{ k}{40}\textbf{w}_{\tau}^{n+6/7},
\end{align}
\begin{align}
s_{f}(\tau_{n+1}) = E-u(\tau_{n+1},x_0),  \quad w(\tau_{n+1},x_0) = -s_{f}(\tau_{n+1});
\end{align}
\begin{equation}
\bar{s}_{f}(\tau_{n+1}) = E-\bar{u}(\tau_{n+1},x_0),  \quad \bar{w}(\tau_{n+1},x_0) = -\bar{s}_{f}(\tau_{n+1}).
\end{equation}
It is worth mentioning that the approximations $s_{f}(\tau_{n+1})$, $\textbf{u}^{n+1} $, and $\textbf{w}^{n+1}$ are fifth-order accurate in time while $\bar{s}_{f}(\tau_{n+1})$, $\bar{\textbf{u}}^{n+1} $, and $\bar{\textbf{w}}^{n+1}$ are fourth-order accurate in time. The numerical approximations $s_{f}(\tau_{n+1})$, $\textbf{u}^{n+1} $, and $\textbf{w}^{n+1}$ represent the solutions of the early exercise boundary, option value and delta sensitivity, respectively. Moreover, both fifth and fourth-order solutions of the option value and delta sensitivity are used to establish an error threshold for the adaptive selection of the optimal time step at each time level as follows:
\begin{equation}\label{otse1}
e_u = ||\textbf{u}^{n+1}-\bar{\textbf{u}}^{n+1}||_{\infty}, \quad e_w = ||\textbf{w}^{n+1}-\bar{\textbf{w}}^{n+1}||_{\infty}, \quad e = \max \lbrace e_u, e_w \rbrace
\end{equation}
\begin{equation}\label{otse3}
    k_{new}^{n+1}=
    \begin{dcases*} 
        \rho \left( \frac{ \epsilon}{e}\right)^{\frac{1}{4}}k_{old}^{n+1}, & if $e$ $<$ $\epsilon$, \\ 
        \rho \left( \frac{ \epsilon}{e}\right)^{\frac{1}{5}}k_{old}^{n+1}, & if $e$ $\geq$ $\epsilon$.
    \end{dcases*}
\end{equation}
Here, $\epsilon$ denotes a tolerance which will be specified in the numerical experiment section. Furthermore, $0 < \rho \leq 1$. 
\section{Numerical Experiment}\label{sec3}
In this section, to verify the performance and efficiency of our proposed method for solving the system of free boundary CEV model, we will use some existing numerical examples to present our results and compare it with some of the exsiting methods. It is worth mentioning that when $\alpha = 0$, we have a standard dividend or non-dividend paying vanilla American options for which we will not consider in this research work. The numerical experiment was carried out on a work-station with a 12th Gen Intel(R) Core(TM) i7-12700H at 2.30 GHz on a 64-bit Windows 11 operating system. For the sake of brevity, we label our methods as given below:
\begin{itemize} 
\item DCU: Uniform fourth-order compact scheme with uniform high-order analytical approximation and fourth-order Runge-Kutta time integration method. 
\item DCSL: Non-equidistant fourth-order compact scheme with staggered high-order analytical approximation and 5(4) Dormand-Prince embedded pairs. 
\end{itemize} 
\noindent In the first example, we consider the example in the work of work of Nicholls and Sward \cite{Nicholls} with the parameter displayed in Table \ref{OptionV}. 

\begin{table}[H]
\center
\caption{Model parameters.}
\label{OptionV}
\begin{tabular}{|lc|}
 \hline 
 Parameters & Values\\ 
 \hline \hline
 $K$    & 9, 10, 11 \\
 $T$    & 0.2 and 0.5 years\\
 $\sigma$ &  0.2 \\
 $S_0$ &  10\\
$\alpha$ & -$ \frac{1}{3}$\\
 $r$ &  0.05 \\
 $x_M$  &  3.00 \\
\hline
 \end{tabular}  
\end{table}
\noindent Here, we investigated the performance of our numerical approximation with the results obtained from the Binomial tree method and the discontinuous Galerkin method of Nicholls and Sward \cite{Nicholls}. We label the discontinuous Galerkin method of Nicholls and Sward \cite{Nicholls} as "DG" and use their results with the maximum accuracy. Here, we want to mainly focus on the impact of the non-uniform and staggered fourth-order numerical schemes we implemented (based on local mesh refinement strategy) on the numerical approximations of the optimal exercise boundary, option value and delta sensitivity. Here, we intend to verify and validate any important gain in computational runtime and solution accuracy we achieved based on the results in Tables \ref{OptionVab224}-\ref{OptionVab99}.
\begin{table}[H]
\center
\caption{American CEV put option value with $\alpha = -\frac{1}{3}$, $\epsilon = 10^{-6}$, and $S_0= 10$.} 
\label{OptionVab224}
\begin{tabular}{|lccc|}
 \hline 
 $K$ & $B_{1000}$ & $B_{5000}$ & DG \cite{Nicholls} \\ 
 \hline \hline
9 &   0.1385 & 0.1385 & 0.1385 \\
10 &   0.4649 & 0.4649 & 0.4649 \\
11 &   1.0894 & 1.0894 & 1.0894 \\
\hline  
 &  & DCU ($\gamma_1 = 2.0, \gamma_2 = 4.0, \gamma_3 = 6.0, \gamma_4 = 8.0$) &   \\
\hline
&   $h =0.1$ & $h =0.06$ & $h =0.03$ \\
 \hline
9 &   0.1482 & 0.1420 & 0.1391 \\
10 &   0.4813 & 0.4706 & 0.4659 \\
11 &   1.1047 & 1.0935 & 1.0902 \\
\hline  
 &  & DCU ($\gamma_1 = 1.0, \gamma_2 = 2.0, \gamma_3 = 3.0, \gamma_4 = 4.0$) &   \\
\hline
&   $h =0.1$ & $h =0.06$ & $h =0.03$ \\
 \hline
9 &   0.1448 & 0.1416 & 0.1393 \\
10 &   0.4733 & 0.4699 & 0.4663 \\
11 &   1.0966 & 1.0931 & 1.0904 \\
\hline
 &  & DCSL ($\gamma_1 = 1.0, \gamma_2=2.0, \gamma_3 = 3.0, \gamma_4 = 4.0$) & \\
 \hline
 &  $h = 0.1$ & $h =0.06$ & $h =0.03$ \\
\hline
9 &   0.1394 & 0.1385 & 0.1384 \\
10 &   0.4656 & 0.4649 & 0.4649 \\
11 &   1.0903 & 1.0894 & 1.0894 \\
\hline
 &  & DCSL ($\gamma_1 = 0.5, \gamma_2=1.0, \gamma_3 = 1.5, \gamma_4 = 2.0$) & \\
 \hline
 &  $h = 0.1$ & $h =0.06$ & $h =0.03$ \\
\hline
9 &   0.1390 & 0.1386 & 0.1385 \\
10 &   0.4649 & 0.4650 & 0.4649 \\
11 &   1.0898 & 1.0895 & 1.0894 \\
\hline
 \end{tabular}  
\end{table}

\begin{table}[H]
\center
\caption{Delta sensitivity for the American CEV model with $\alpha = -\frac{1}{3}$, $\epsilon = 10^{-6}$, and $S_0= 10$.} 
\label{OptionVabvvb}
\begin{tabular}{|lccc|}
\hline  
  $K$ &  & DCU ($\gamma_1 = 2.0, \gamma_2 = 4.0, \gamma_3 = 6.0, \gamma_4 = 8.0$) &   \\
\hline
&   $h =0.1$ & $h =0.06$ & $h =0.03$ \\
 \hline
9 &   -0.1839 & -0.1800 & -0.1779 \\
10 &   -0.4473 & -0.4424 & -0.4410 \\
11 &   -0.7622 & -0.7569 &  -0.7587 \\
\hline  
 &  & DCU ($\gamma_1 = 1.0, \gamma_2 = 2.0, \gamma_3 = 3.0, \gamma_4 = 4.0$) &   \\
\hline
&   $h =0.1$ & $h =0.06$ & $h =0.03$ \\
 \hline
9 &   -0.1799 & -0.1796 & -0.1780 \\
10 &   -0.4419 & -0.4419 & -0.4409 \\
11 &   -0.7642 & -0.7568 &  -0.7584 \\
\hline
 &  & DCSL ($\gamma_1 = 1.0, \gamma_2=2.0, \gamma_3 = 3.0, \gamma_4 = 4.0$) & \\
 \hline
 &  $h = 0.1$ & $h =0.06$ & $h =0.03$ \\
\hline
9 &   -0.1775 & -0.1775 & -0.1775 \\
10 &   -0.4402 & -0.4409 & -0.4409 \\
11 &   -0.7574 & -0.7594 &  -0.7594 \\
\hline
 &  & DCSL ($\gamma_1 = 0.5, \gamma_2=1.0, \gamma_3 = 1.5, \gamma_4 = 2.0$) & \\
 \hline
 &  $h = 0.1$ & $h =0.06$ & $h =0.03$ \\
\hline
9 &   -0.1773 & -0.1775 & -0.1775 \\
10 &   -0.4404 & -0.4409 & -0.4409 \\
11 &   -0.7582 & -0.7595 &  -0.7594 \\
\hline
 \end{tabular}  
\end{table}

\begin{table}[H]
\center
\caption{Runtime in seconds for the American CEV model with $\alpha = -\frac{1}{3}$, $\epsilon = 10^{-6}$, and $S_0= 10$.} 
\label{OptionVab99}
\begin{tabular}{|lccc|}
\hline  
$K$ &  & DCU ($\gamma_1 = 2.0, \gamma_2 = 4.0, \gamma_3 = 6.0, \gamma_4 = 8.0$) &   \\
\hline
&   $h =0.1$ & $h =0.075$ & $h =0.05$ \\
 \hline
9 &   0.079 & 0.088 & 0.155 \\
10 &   0.074 & 0.086 & 0.153 \\
11 &   0.076 & 0.086 & 0.153 \\
\hline  
 &  & DCU ($\gamma_1 = 1.0, \gamma_2 = 2.0, \gamma_3 = 3.0, \gamma_4 = 4.0$) &   \\
\hline
&   $h =0.1$ & $h =0.075$ & $h =0.05$ \\
 \hline
9 &   0.081 & 0.106 & 0.274 \\
10 &   0.079 & 0.105 & 0.266 \\
11 &   0.082 & 0.103 & 0.262 \\
\hline
 &  & DCSL ($\gamma_1 = 1.0, \gamma_2=2.0, \gamma_3 = 3.0, \gamma_4 = 4.0$) & \\
 \hline
 &  $h = 0.1$ & $h =0.075$ & $h =0.05$ \\
\hline
9 &  0.105 & 0.173 & 0.787 \\
10 &   0.108 & 0.177 & 0.846 \\
11 &   0.111 & 0.182 & 0.921 \\
\hline
 &  & DCSL ($\gamma_1 = 0.5, \gamma_2=1.0, \gamma_3 = 1.5, \gamma_4 = 2.0$) & \\
 \hline
 &  $h = 0.1$ & $h =0.06$ & $h =0.03$ \\
\hline
9 &   0.113 & 0.214 & 1.085 \\
10 &   0.115 & 0.199 & 1.006 \\
11 &   0.111 & 0.206 & 0.958 \\
\hline
 \end{tabular}  
\end{table}
\noindent Our observation is that when we refine the space grid in the locality of the left corner point and further use grid points very close to the left boundary for the approximation of the early exercise boundary based on the staggered representation, we only need coarse grids to achieve very reasonable accuracy. Moreover, we achieve reasonable accuracy with very little computational runtime in seconds even though we approximate the early exercise boundary, option value and delta sensitivity simultaneously. Furthermore, because the Binomial tree solution which serves as a benchmark for this example was obtained with a maximum number of 5000 steps, we envisaged the need to further use a very small step size to compute the numerical solutions of the option value, delta sensitivity and optimal exercise boundary and present result in 9 decimal places as described in Table \ref{OptionVabF} below.
\begin{table}[H]
\center
\caption{Numerical approximation with DCSL ($\gamma_1 = 0.5, \gamma_2 = 1, \gamma_3 = 1.5, \gamma_4 = 2$), $h=0.01$, $\alpha = -\frac{1}{3}$, $\epsilon = 10^{-6}$, and $S_0= 10$.} 
\label{OptionVabF}
\begin{tabular}{|lcc|}
\hline
$K$  & Option value & Delta sensitivity \\
 \hline
9   & 0.13845314 & -0.17749883 \\
10  & 0.46492683 & -0.44091309 \\
11  & 1.08943407 & -0.75936161 \\
\hline
 \end{tabular}  
\end{table}
\noindent Furthermore, we computed the convergence rate of our proposed method with the example in Table \ref{OptionV} above. Because of the adaptive nature of 5(4) embedded pairs, we rather used the fourth-order Runge-Kutta method and non-equidistant fourth-order compact finite difference scheme with the staggered high-order analytical approximation to compute only the convergence rate in space for the optimal exercise boundary. The convergence rate result with the optimal exercise boundary is listed in Table \ref{OptionVabc}. From Table \ref{OptionVabc}, we observed a reasonable convergence rate result. However, we achieve third-order convergence rate result with the non-uniform compact finite difference scheme instead of fourth-order convergence rate result.
\begin{table}[H]
\center
\caption{Maximum errors and convergence rates with the optimal exercise boundary ($T = 0.2$, $k = 2.5 \times 10^{-7}$, $\alpha = -\frac{1}{3}$, and $K= 9$). Here, we use DCSL ($\gamma_1 = 1, \gamma_2 = 2, \gamma_3 = 3, \gamma_4 = 4$).} 
\label{OptionVabc}
\begin{tabular}{|lcc|}
\hline
 $h$ &   Maximum errors & Convergence rate\\
 \hline
0.2 &   &  \\
0.1 &   0.09707782 &  \\
0.05 &   0.01077838 & 3.171 \\
0.025 &     0.00117112 & 3.202 \\
0.0125 &    0.00011777 & 3.314 \\
\hline
 \end{tabular}  
\end{table}
\noindent In the second example, we consider the example in the work of Nunes \cite{Nunes} with the following options parameters listed in Table \ref{OptionVaa22x}.
\begin{table}[H]
\center
\caption{Option parameters.}
\label{OptionVaa22x}
\begin{tabular}{|lc|}
 \hline 
 Parameter & Values\\ 
 \hline \hline
 $K$    & 80, 90, 100, 110, 120\\
 $T$    & 0.5 years\\
 $S_0$    & 100 \\
 $\sigma$ &  0.3  \\
$\alpha$ & $\frac{1}{2}$ \\
 $r$ & 0.07 \\
 $x_M$  &  3.00 \\
\hline
 \end{tabular}  
\end{table}
\noindent It is worth mentioning that Nunes \cite{Nunes} considered the following stochastic differential equation.
\begin{equation} 
dS_t= \mu S_tdt + \sigma S^{\frac{\beta}{2} }_tdW(t). 
\end{equation}
However, we consider
\begin{equation} 
dS_t= \mu S_tdt + \sigma S^{\alpha+1}_t dW(t). \quad  \alpha \in \mathbb{Q}.
\end{equation}
The discrepancy given in the work of Nunes \cite{Nunes} can be corrected in our work for which we did accordingly below. if we set  $\alpha+1 = \beta/2$, it implies that $\alpha = \beta/2-1$.
Nunes \cite{Nunes} used $\beta= 3$ in their work which corresponds to $\alpha = 0.5$. Hence, the reason for slight change in the $\alpha$ in Table \ref{OptionVaba} to be consistent with the one presented in the work of Nunes \cite{Nunes}. The authors used the result obtained from a second-order Crank-Nicolson method with 15,000 time intervals
and 10,000 space steps as the benchmark value. Here we will verify the consistency in our result with the benchamrk value and the method suggested by Detemple and Tian \cite{Detemple} labelled as "DT" and Kim and Yu \cite{Kim} labelled as "KY" which are listed in the work of Nunes \cite{Nunes}. The numerical approximation and the computational runtime results are listed in Tables \ref{OptionVaba}-\ref{OptnVb} below.
\begin{table}[H]
\center
\caption{American CEV put option value with $\alpha = \frac{1}{2}$, $\epsilon = 10^{-6}$, and $S_0= 100$.} 
\label{OptionVaba}
\begin{tabular}{|lccc|}
 \hline 
 $K$ & Crank-Nicholson & DT & KY \\ 
 \hline \hline
80 &   0.852 & 0.851 & 0.851 \\
90 &   2.969 &  2.968 & 2.969 \\
100 &  7.060 & 7.060 & 7.060 \\
110  &  13.175 &  13.174 & 13.172 \\
120 &  20.992 & 20.989 & 20.993 \\
\hline  
 &  & DCSL ($\gamma_1 = 0.5, \gamma_2=1, \gamma_3 = 1.5, \gamma_4 = 2$) & \\
 \hline
$h$ &  0.1 & 0.075 & 0.03 \\
\hline
80 &  0.854 & 0.853 & 0.852 \\
90 &  2.971 & 2.969 & 2.969 \\
100 &  7.061 & 7.061 & 7.060 \\
110  & 13.174 & 13.176 & 13.175 \\
120 &  20.991 & 20.992 & 20.992 \\
\hline
 \end{tabular}  
\end{table}

\begin{table}[H]
\center
\caption{Delta sensitivity for the American CEV model with $\alpha = 0.5$, $\epsilon = 10^{-6}$, and $S_0= 100$.} 
\label{OptionVab}
\begin{tabular}{|lccc|}
\hline  
 &  & DCSL ($\gamma_1 = 0.5, \gamma_2=1, \gamma_3 = 1.5, \gamma_4 = 2$) & \\
 \hline
$h$ &  0.1 & 0.075 & 0.03 \\
\hline
80 &  -0.079 & -0.079 & -0.079 \\
90 &  -0.214 & -0.214 & -0.214 \\
100 &  -0.402 & -0.402 & -0.402  \\
110  & -0.602 & -0.603 & -0.602 \\
120 &  -0.785 & -0.785 & -0.785 \\
\hline
 \end{tabular}  
\end{table}

\begin{table}[H]
\center
\caption{Runtime in seconds for the American CEV model with $\alpha = 0.5$, $\epsilon = 10^{-6}$, and $S_0= 100$.} 
\label{OptnVb}
\begin{tabular}{|lccc|}
\hline  
 &  & DCSL ($\gamma_1 = 0.5, \gamma_2=1, \gamma_3 = 1.5, \gamma_4 = 2$) & \\
 \hline
$h$ &  0.1 & 0.075 & 0.03 \\
\hline
80 & 0.144 & 0.195 & 1.090 \\
90 &  0.143 & 0.200 & 1.218 \\
100 &  0.153 & 0.210 & 1.306  \\
110  & 0.177 & 0.223 & 1.388 \\
120 &  0.293 & 0.231 & 1.467 \\
\hline
 \end{tabular}  
\end{table}
\noindent From Table \ref{OptionVaba}, we observe a reasonable match between our result and the well performing results presented in Table \ref{OptionVaba}. Moreover, the advantage of our proposed method is that the option value, delta sensitivity and the optimal exercise boundary are approximated with a reasonable accuracy using very small step sizes when compared with the benchmark method. This is largely due to the implementation of fourth-order compact finite difference scheme in space and fifth-order 5(4) embedded pair for time integration. Furthermore, several grid refinements for the approximations of the early exercise boundary, option value and the delta sensitivity further contributed to the overall enhancement of our proposed method.  Moreover, only little computational runtime is required to achieve this reasonable accuracy as shown in Table \ref{OptnVb}.\\

\noindent In this last example, we investigate the American style CEV model with integer elasticity $\alpha$. Consider the example in the work of Lee et al. \cite{Lee} with the options parameter listed in Table \ref{OptionVaa223b}.
\begin{table}[H]
\center
\caption{Option parameters.}
\label{OptionVaa223b}
\begin{tabular}{|lc|}
 \hline 
 Parameter & Values\\ 
 \hline \hline
 $K$    & 35,40,45\\
 $T$    & 3.0 years\\
 $S_0$    & 40 \\
 $\sigma$ &  0.3  \\
$\alpha$ &  -1.0 \\
 $r$ & 0.05 \\
 $x_M$  &  3.00 \\
\hline
 \end{tabular}  
\end{table}
\noindent Here, we first present the result of the option value and delta sensitivity and verify our result with the Crank-Nicholson finite difference (CFDM), the method of Lee \cite{Lee} for which we label as TM, artificial boundary condition (ABC) method of Wong and Zhao \cite{Wongb}, and the Binomial tree method with 3000 steps presented in the work of. For the finite difference method of Lee \cite{Lee} and ABC method of Wong and Zhao \cite{Wongb}, we use the numerical approximations with the maximum accuracy as presented in the work of Lee \cite{Lee}. The results are displayed in Table \ref{OptVab224}. Similar to our observations with DCSL in the previous examples, reasonable solution accuracy is achieved with a very large step sizes, an important gain that could reduce space complexity when extended to high-dimensional pricing problems.
\begin{table}[H]
\center
\caption{American CEV put option value with $\alpha = -1$, $\epsilon = 10^{-6}$, and $S_0 = 40$.} 
\label{OptVab224}
\begin{tabular}{|lccc|}
 \hline 
 $K$ & CFDM & TM\cite{Lee}  & ABC \cite{Wongb} \\ 
 \hline 
 35 &   4.040 & 4.034 & 4.040 \\
40 &   5.792 & 5.791 & 5.792 \\
45 &   8.113 & 8.113 & 8.113 \\
\hline
 &  & DCSL ($\gamma_1 = 0.5, \gamma_2=1.0, \gamma_3 = 1.5, \gamma_4 = 2.0$) & \\
 \hline
 &  $h = 0.1$ & $h =0.06$ & $h =0.03$ \\
\hline
35 &   4.041 & 4.041 & 4.041 \\
40 &   5.792 & 5.792 & 5.792 \\
45 &  8.114 & 8.114 & 8.114 \\
\hline
 \end{tabular}  
\end{table}
\begin{figure}[H]
    \centering
    {\includegraphics[width=1.0\textwidth]{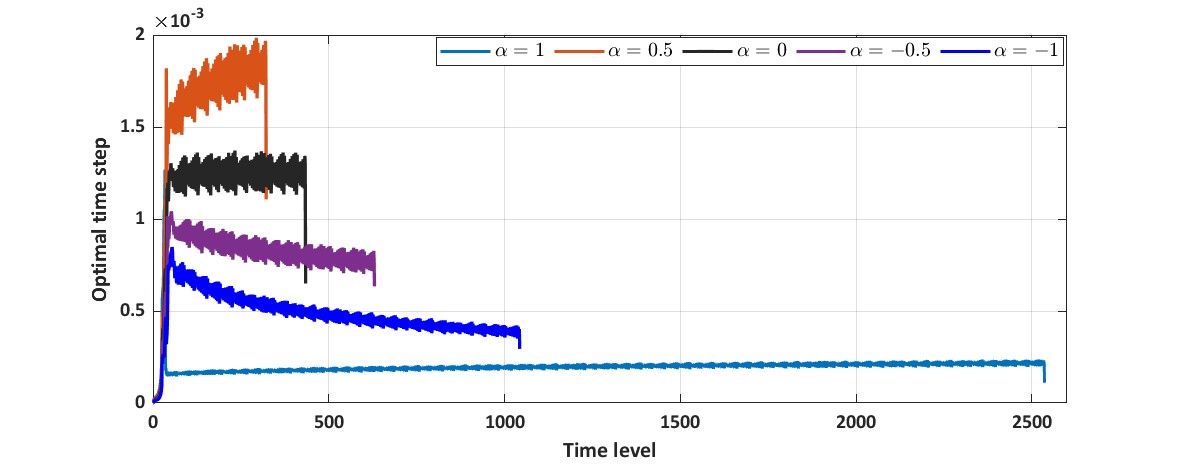}} 
    \caption{Optimal time stepping for each time level with varying $\alpha$.}
    \label{fig:Local1}
\end{figure}
\noindent Furthermore, we present the plot profile of the optimal time step with DCSL at each time level and observe how the integer elasticity impacts on the selection of the optimal time step at each time level based on the 5(4) Dormand-Prince embedded pairs. For a fixed step size, our observation is that as the absolute value of $\alpha$ increases, the optimal time step for each time level reduces, hence, requiring more computational runtime and number of time levels.\\

\noindent Finally, we vary the parameters of the 5(4) Dormand-Prince embedded pairs and further observe how they impact on the solution accuracy of our numerical approximation. Precisely, we vary $\epsilon$ and $\rho$. The results are listed in Tables \ref{OptVb224} and \ref{OptVb225}. We observe that the solution accuracy is not susceptible to  $\rho$ and vary substantially with $\epsilon$.
\begin{table}[H]
\center
\caption{American CEV put option value with varying $\epsilon$ ($\alpha = -1, h=0.05, \rho = 0.9$).} 
\label{OptVb224}
\begin{tabular}{|lccc|}
 \hline 
$K/S_0$ &  & DCSL ($\gamma_1 = 0.5, \gamma_2=1.0, \gamma_3 = 1.5, \gamma_4 = 2.0$) & \\
 \hline \hline
 &  $ \epsilon = 10^{-4}$ & $\epsilon = 10^{-5}$ & $\epsilon = 10^{-6}$ \\
\hline
35/40 &   4.3855 & 4.0429 & 4.0408 \\
\hline
 \end{tabular}  
\end{table}

\begin{table}[H]
\center
\caption{American CEV put option value with varying $\rho$ ($\alpha = -1, h=0.05, \epsilon = 10^{-6}$).} 
\label{OptVb225}
\begin{tabular}{|lccc|}
 \hline 
  $K/S_0$ &  & DCSL ($\gamma_1 = 0.5, \gamma_2=1.0, \gamma_3 = 1.5, \gamma_4 = 2.0$) & \\
 \hline \hline
 &  $ \rho = 0.3$ & $ \rho = 0.6$ & $ \rho = 0.9$ \\
\hline
35/40 &   4.0408 & 4.0408 & 4.0408 \\
\hline
 \end{tabular}  
\end{table}
\section{Conclusion} \label{sec4}
We have proposed a high-order uniform and staggered numerical schemes on a locally adapted space grid with adaptive time stepping for approximating the American style CEV model. We fully demonstrated in our work that even though the uniform space grid scheme performed reasonable well, formulation of staggered implementation substantially improved accuracy in coarse grids and reduced computational runtime. Our model is highly nonlinear and a bit complex and involves a system of free boundary partial differential equations. Hence, our numerical formulation is extendible and presents an efficient and alternative numerical method for solving other linear and nonlinear system of PDEs with or without free boundary in the field of finance and other disciplines. Furthermore, with the non-uniform and staggered implementation, we have shown that we can achieve very reasonable accuracy with coarse grids (large step sizes), hence improving space complexity and can further be extendible to a reasonable high dimensional problems.\\   

\noindent In our future work, we hope to extend our proposed method for solving regime switching CEV model which was formulated and described in the work of Elliot et al. \cite{Elliott}. In our context, it will involve multi-free boundaries. We hope to leverage the benefit of our high-order staggered analytical approximation, local mesh refinement and adaptive strategy to solve such free boundary regime switching options with very coarse grids. This will allow us to approximate multiple regime problems with reasonable solution accuracy in very coarse grids. Furthermore, it is also in our interest to understand if we can stretch our numerical methodology and model transformation procedure to further solve other related free boundary local volatility models like the VIX model, nonlinear mean reverting models, and other local volatility models involving Pearson processes as presented in the work of Leonenko and Phillips \cite{Leonenko}.

\section*{Acknowledgement}
The first author is funded in part by an NSERC Discovery Grant.

%

\end{document}